\documentclass[arxiv,preprint]{imsart}

\RequirePackage[OT1]{fontenc}
\RequirePackage{amsmath,amsthm,amssymb}
\usepackage{graphicx,psfrag,epsf}
\usepackage{enumerate}
\RequirePackage{natbib}
\usepackage{url} 

\usepackage{amsfonts}
\usepackage{array}
\usepackage{color}
\usepackage{xcolor}
\usepackage{bm}
\usepackage{enumerate}
\usepackage{float}
\usepackage[section]{placeins}
\usepackage{multirow}
\usepackage{picinpar}
\usepackage{pifont}
\usepackage{algorithm}
\usepackage{algorithmic}

\newtheorem{lemma}{Lemma}[section]
\newtheorem{corollary}[lemma]{Corollary}
\newtheorem{condition}[lemma]{Condition}
\newtheorem{theorem}[lemma]{Theorem}
\newtheorem{remark}[lemma]{Remark}

\newcommand{\var}{\textsf{Var}}

\newcommand{\varP}{\textsf{Var}}

\newcommand{\cov}{\textsf{Cov}}
\newcommand{\cor}{\textsf{Cor}}

\newcommand{\covP}{\textsf{Cov}}

\newcommand{\I}{\mathbf{I}}
\newcommand{\E}{\textsf{E}}

\newcommand{\EP}{\textsf{E}}

\newcommand{\PP}{\textsf{P}}

\newcommand{\pp}{\textsf{P}}

\newcommand{\CalD}{\mathcal{D}}
\newcommand{\e}{\mathcal{E}}

\newcommand{\V}{\mathcal{V}}

\newcommand{\ba}{\textbf{a}}
\newcommand{\bb}{\textbf{b}}
\newcommand{\bc}{\textbf{c}}
\newcommand{\bd}{\textbf{d}}
\newcommand{\be}{\textbf{e}}

\arxiv{arXiv:0000.0000}

\startlocaldefs
\numberwithin{equation}{section}
\theoremstyle{plain}
\endlocaldefs

\graphicspath{{image/}}

\begin{document}

\begin{frontmatter}

\title{Graph-based two-sample tests for data with repeated observations}
\runtitle{Graph-based two-sample tests}

\begin{aug}
\author{\fnms{Jingru} \snm{Zhang}\ead[label=e1]{jingruzhang@pku.edu.cn}}
\and
\author{\fnms{Hao} \snm{Chen}\ead[label=e2]{hxchen@ucdavis.edu}}
\runauthor{Zhang and Chen}
\address{Beijing International Center for Mathematical Research \\
and University of California, Davis}
\end{aug}

\begin{abstract}
In the regime of two-sample comparison, tests based on a graph constructed on observations by utilizing similarity information among them is gaining attention due to their flexibility and good performances for high-dimensional/non-Euclidean data. 
However, when there are repeated observations, these graph-based tests could be problematic as they are versatile to the choice of the similarity graph. We propose extended graph-based test statistics to resolve this problem.  The analytic $p$-value approximations to these extended graph-based tests are derived to facilitate the application of these tests to large datasets. The new tests are illustrated in the analysis of a phone-call network dataset.  All tests are implemented in an \texttt{R} package \texttt{gTests}.
\end{abstract}


\end{frontmatter}

\section{Introduction}
\label{sec:intro}

Two-sample comparison is a fundamental problem in statistics and has been extensively studied for univariate data and low-dimensional data. The testing problem for high-dimensional data and non-Euclidean data, such as network data, is gaining more and more attention in this big-data era.  In the parametric domain, for multivariate data, many endeavors have been made in testing whether the means are the same (see for examples \cite{bai1996, srivastava2008, chensongxi2010, tony2014, xu2016}) and whether the covariance matrices are the same (see for examples \cite{schott2007, srivastava2010, li2012, tony2013, xia2015}).   Parametric tests were proposed for non-Euclidean data as well.  For example, \cite{tang2017semiparametric} proposed a test for two random dot product graphs based on the spectral decomposition of the adjacency matrix.  These parametric methods provide useful tools, but they are often restrictive and not robust enough if model assumptions are violated. 

In the nonparametric domain, efforts had been made in extending the Kolmogorov-Smirnov test, the Wilcoxon rank test, and the Wald-Wolfowitz runs test to high-dimensional data (see \cite{chen2017new} for a review).  Among these efforts, the first practical test was proposed by \cite{friedman1979} as an extension of the Wald-Wolfowitz runs test to multivariate data.  They pool the observations from the two samples together and construct a minimum spanning tree (MST), which is a spanning tree that connects all observations with the sum of the distances of the edges in the tree minimized.  They then count the number of edges in the MST that connects observations from different samples, and reject the null hypothesis of equal distribution if this count is significantly \emph{smaller} than its expectation under the null hypothesis.   
This test later has been extended to other similarity graphs where observations closer in distance are more likely to be connected than those farther in distance, such as the minimum distance pairing (MDP) where the observations are paired in such a way that the sum of the distances within pairs is minimized \citep{rosenbaum2005}, and the nearest neighbor graph (NNG) where each observation connects to its nearest neighbor \citep{schilling1986, henze1988}.  We call this type of tests the \emph{edge-count test} for easy reference.
Recently, a generalized edge-count test and a weighted edge-count test were proposed to address the problems of the edge-count test under scale alternatives and under unequal sample sizes \citep{chen2017new, chen2018weighted}. Since these tests and the edge-count test are all based on a similarity graph, we call them the \emph{graph-based tests}.


The graph-based tests have many advantages: They can be applied to data with arbitrary dimension and to non-Euclidean data, and exhibit high power in detecting a variety of differences in distribution -- they have higher power than the likelihood-based tests when the dimension of the data is moderate to high for practical sample sizes, ranging from hundreds to millions.  In all the works relating the graph-based tests, the authors also provided analytic formulas to approximate the $p$-values of the corresponding test statistics, making the tests easy off-the-shelf tools for two-sample comparison in modern applications. However, the graph-based tests could be problematic for data with repeated observations. All these tests rely on a similarity graph constructed on the observations. 
When there are repeated observations, the similarity graph is not uniquely defined based on common optimization criteria, such as the MST or the MDP. Indeed, several graphs could be equally ``optimal'' in terms of the criterion.  

To illustrate this problem, we use a phone-call network dataset analyzed in both \cite{chen2017new} and \cite{chen2018weighted}. This dataset has 330 networks, corresponding to 330 consecutive days, respectively. Each network represents the phone-call activity among the same group of people on a particular day (a more detailed description of this dataset see in Section \ref{sec:network}). In both papers, the authors tested whether the distribution of phone-call networks on weekdays is the same as that on weekends. The distance between two networks is defined as the number of different edges between them. In this dataset, phone-call networks on some days are the same and the distance matrix on the distinct networks has ties. According to their results, the 9-MST\footnote{A $k$-MST is the union of the $1$st,$\cdots,k$th MSTs, where the 1st MST is the MST and the $j$th ($j>1$) MST is a spanning tree that connects all observations such that the sum of the edges in the tree is minimized under the constraint that it does not contain any edge in the $1$st,$\cdots,(j-1)$th MSTs.} was a good choice for the similarity graph. However, the 9-MST is not uniquely defined due to the repeated observations (networks) and the ties in the distance matrix.
We randomly selected four such 9-MSTs and the results of the generalized edge-count test ($S$) and the weighted edge-count test ($Z_w$) under each of the 9-MSTs are listed in Table \ref{Tab:ex}. 
We see that the test statistics based on different 9-MSTs vary a lot and the $p$-values could be very small under some choices of 9-MSTs but very large under some other choices, leading to completely different conclusions.


\begin{table}[!htbp]
  \centering
  \caption{The test statistics and their corresponding $p$-values (in parentheses, bold for those $<0.05$) of the generalized edge-count test ($S$) and the weighted edge-count test ($Z_w$) under four 9-MSTs on the phone-call network data. }\label{Tab:ex}
  \begin{tabular}{|c|cccc|}
    \hline
    MST & \#1 & \#2 & \#3 & \#4 \\ \hline
    $S$ & 6.86 (\textbf{0.032}) & 3.92 (0.141) & ~7.89 (\textbf{0.019}) & 3.90 (0.142) \\ \hline
     $Z_w$ & 2.61 (\textbf{0.004}) & 1.95 (\textbf{0.025}) & -1.13 (0.871) & 0.26 (0.396) \\ \hline
  \end{tabular}
\end{table}

In this work, we seek ways to effectively summarize the tests over these equally ``optimal" similarity graphs.  As we will show in Section \ref{sec:optimalG}, it is easy to have more than a million equally optimal similarity graphs when there are repeat observations, so manually examining the results from each of these graphs is usually not feasible.  This work borrows ideas from \cite{chen2013}, where the authors considered two ways in summarizing the \emph{original} edge-count test statistic over equally optimal similarity graphs.  However, extending the \emph{generalized} edge-count test directly as the original edge-count test is technically intractable due to the quadratic terms in the test statistics (details see in Section \ref{sec:extended}).  To get around this issue, we first summarize the basic quantities contributing to the graph-based tests over equally optimal similarity graphs, which we refer to as the \emph{extended basic quantities}.  We then construct extended generalized and weighted edge-count tests based on these extended basic quantities so that they can handle data with repeated observations.  In particular, we proved the following results: 
\begin{enumerate}[(1)]
	\item The extended weighted edge-count test statistic constructed in this way adopts the same weights as the weighted edge-count test to resolve the variance boosting problem of the edge-count test when the sample sizes of the two samples are different.
	\item The extended generalized edge-count test statistic constructed in this ways is composed of two asymptotically independent quantities. 
\end{enumerate}
Based on (2), we further study an extended max-type edge-count test that builds upon the two asymptotically independent quantities.  We also derive analytic $p$-value approximations for all the new test statistics, making them fast applicable to real datasets.  The tests are all implemented in an \texttt{R} package \texttt{gTests}.

The rest of the paper is organized as follows. Section \ref{sec:notation} provides notations used in the paper and preliminary setups.  Section \ref{sec:extended} discusses in details the extended weighted, generalized, and max-type edge-count tests.  The performance of these new tests are examined in Section \ref{sec:performance} and their asymptotic properties are studied in Section \ref{sec:asym}. Section \ref{sec:network} illustrate the new tests in the analysis of the phone-call network dataset.

\section{Notations and preliminary setups}
\label{sec:notation}

\subsection{Basic notations} 
\label{sec:basic}
For data with repeated observations, assume that there are $K$ distinct values and we index them by $1,2,\cdots,K$.  Throughout the paper, we use the notations summarized in Table \ref{tab:notation}. 
\begin{table}[!htbp]
  \centering
  \caption{Data with repeated observations summarized by distinct values.}\label{tab:notation}
  \begin{tabular}{|cccccc|}
    \hline
    Distinct value index & 1 & 2 & $\cdots$ & K & Total\\ \hline
    Sample 1 & $n_{11}$ & $n_{12}$ & $\cdots$ & $n_{1K}$ & $n_1$ \\ \hline
    Sample 2 & $n_{21}$ & $n_{22}$ & $\cdots$ & $n_{2K}$ & $n_2$ \\ \hline
    Total & $m_1$ & $m_2$ & $\cdots$ & $m_K$ & N \\
    \hline
  \end{tabular} \\
  \vspace{0.2cm}
  Here, $m_i=n_{1i}+n_{2i},~i=1,\cdots,K;~n_i=\sum_{k=1}^{K}n_{ik},~i=1,2;~N=n_1+n_2.$
\end{table}

Let $\{d(i,j):i,j=1,\cdots,K\}$ be the distance matrix on the distinct values, with $d(i,j)$ being the distance between values indexed by $i$ and $j$.  For an undirected graph $G$, let $|G|$ be the number of edges in $G$. For any node $i$ in the graph $G$, let $\e^G_i$ be the set of edges in $G$ that contain node $i$, and $\mathcal{V}_i^G$ be the set of nodes in $G$ that connect to node $i$. 

Since there is no distributional assumption on the data, we work under the permutation null distribution, which places $1/\binom{N}{n_1}$ probability on each of the $\binom{N}{n_1}$ ways of assigning the sample labels such that sample 1 has $n_1$ observations.  Without further specification, we use $\E$, $\var$, $\cov$, $\cor$ to denote the expectation, variance, covariance and correlation under the permutation null distribution.

\subsection{Similarity graphs on observations}
\label{sec:optimalG}
Let $C_0$ be a similarity graph constructed on the distinct values. It could be the MST, the MDP, or the NNG on the distinct values if it can be uniquely defined. 

If the common optimization rules do not result in an unique solution, we adopt the same treatment as in \cite{chen2013} by using the union of all MSTs.  Figure \ref{fig:example} is a simple example.  It can be shown that this union of all MSTs on the distinct values can be obtained through Algorithm \ref{NNL}.  For example, for the data in Figure \ref{fig:example}, distinct values \ba ~and \bb, \ba ~and \bc, \bb ~and \bc, \bd ~and \be ~are connected in the first step, then \bb ~and \bd, \bc ~and \be ~are connected in the second step. We call this graph the nearest neighbor link (NNL).  If one wants denser graphs, $k$-NNL could be considered, which is the union of the 1st,$\cdots,k$th NNLs, where the $j$th ($j>1$) NNL is a graph generated by Algorithm \ref{NNL} subject to the constraint that this graph does not contain any edge in the 1st,$\cdots,(j-1)$th NNLs. 

\begin{figure}
  \centering
  \includegraphics[width=1\textwidth]{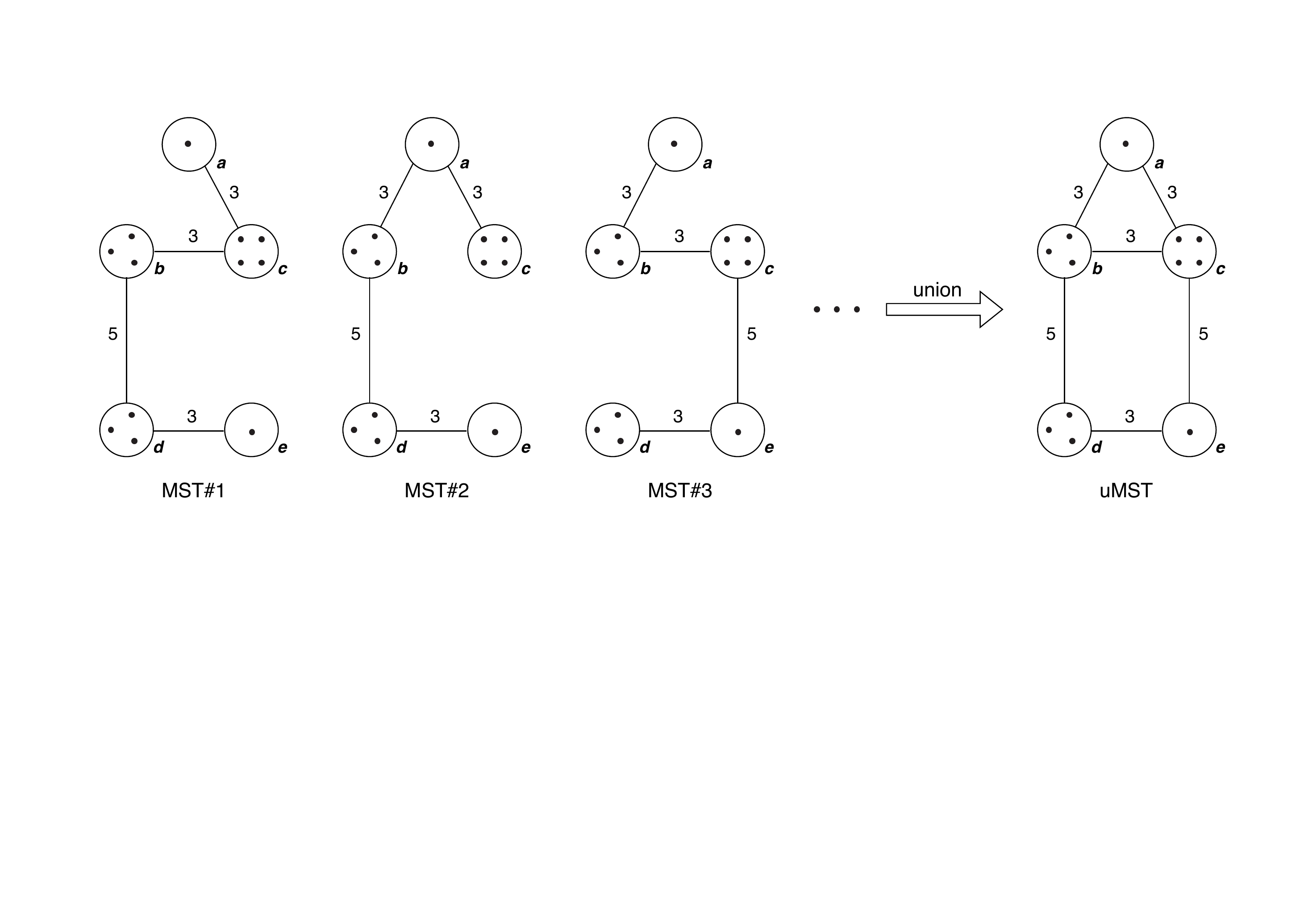}
  \vspace{-4.7cm}
  \caption{There are five distinct values (\ba, \bb, \bc, \bd, \be) denoted by the circles. For some distinct values, there are more than one observations, denoted by the more than one point in the circle. The distance between the distinct values are denoted on the edges.  It is clear that there are six MSTs on the distinct values (three of them are presented on the left) and the last plot is the union of the six MSTs on the distinct values. }\label{fig:example}
\end{figure}

\begin{algorithm}
  \caption{Generate a NNL}
  \label{NNL}
  \begin{algorithmic}
     \STATE For each distinct value indexed by $i~(i=1,\cdots,K)$, let $d_{\min}(i) = \min\{d(i,j): j\neq i\}$ and $N(i)=\{j: d(i,j)=d_{\min}(i)\}$.  Connect $i$ to each element in $N(i)$. 
    \WHILE{Not all distinct values are in one component} 
    \STATE Let $\mathcal{U}$ be one component, let $d_{\min}(\mathcal{U}) = \min\{d(i,j): i\in\mathcal{U}, j\notin\mathcal{U}\} $ and $N(\mathcal{U}) = \{(i,j):d(i,j) = d_{\min}(\mathcal{U}),i\in\mathcal{U}, j\notin\mathcal{U}\}$.  Connect each pair of distinct values indexed by $i$ and $j$ if $(i,j)\in N(\mathcal{U})$.
   \ENDWHILE
  \end{algorithmic}
\end{algorithm}

Then, a graph on observations initiated from $C_0$ can be defined in the following way: First, for each pair of edges  $(i,j)\in C_0$, randomly choose an observation with value indexed by $i$ and another observation with value indexed by $j$, connect these two observations; then, for each $i$, if there are more than one observation with value indexed by $i$, connect these observations by a spanning tree (any edge in such a spanning tree has distance 0).  
We denote all the these graphs by $\mathcal{G}_{C_0}$.

For the example in Figure \ref{fig:example}, since the MST on the distinct values is not uniquely defined, let $C_0$ be the NNL. There are only 5 distinct values and 6 edges on $C_0$. However, there are $15,552(=1^2\cdot 3^3\cdot 4^3 \cdot 3^2 \cdot 1^2)$ different ways in assigning the 6 edges in $C_0$ to corresponding observations in each circle. In addition, by Cayley's lemma, for the observations belonging to the same distinct value, there are 1, 3, 16, 3 and 1 spanning trees, respectively. Therefore, we have $2,239,488(= 15,552\times 3\times 16\times 3)$ graphs in $\mathcal{G}_{C_0}$.  Figure \ref{fig:plot-embedding} plots four of these graphs for illustration.  

\begin{figure}
  \centering
  \includegraphics[width=\textwidth]{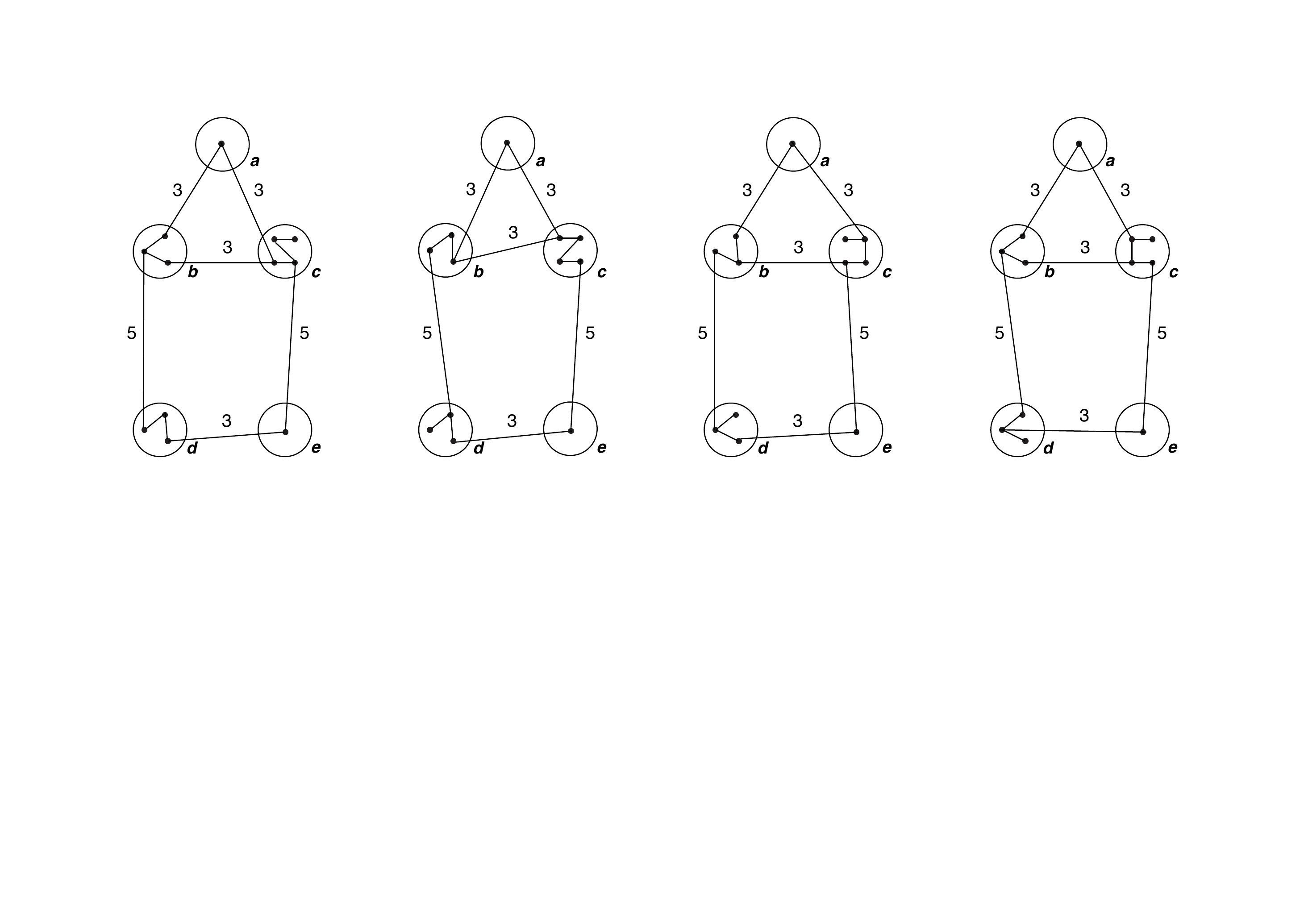}
  \vspace{-5cm}
  \caption{Four graphs, out of 2,239,488, on observations initiated from the NNL on distinct values.}\label{fig:plot-embedding}
\end{figure}


\subsection{Basic quantities in the graph-based tests}
\label{sec:basic}
For any graph $G$, let $R_{0,G}$ be the number of edges in $G$ that connect observations from different samples, $R_{1,G}$ be the number of edges in $G$ that connect observations from sample 1, and $R_{2,G}$ be that for sample 2. Here, $R_{0,G}$ is the test statistic for the original edge-count test.
In \cite{chen2017new}, the authors noticed that, the edge-count test ($R_{0,G}$) has low or even no power for scale alternatives when the dimension is moderate to high unless the sample size is extremely large due to the curse-of-dimensionality. To solve this problem, they considered the numbers of within-sample edges of the two samples separately and proposed the following generalized edge-count statistic
\begin{equation}\label{eq:S}
S_G = \left( \begin{array}{c}
	R_{1,G} - \E(R_{1,G}) \\ R_{2,G} - \E(R_{2,G}) 
\end{array} \right)^T \mathbf\Sigma_G^{-1}\left( \begin{array}{c}
	R_{1,G} - \E(R_{1,G}) \\ R_{2,G} - \E(R_{2,G}) 
\end{array} \right),
\end{equation}
where $\mathbf\Sigma_{G} = \var(\binom{R_{1,G}}{R_{2,G}})$.
 
Both the edge-count test and the generalized edge-count test are suggested to perform on a similarity graph that is denser than the MST, such as 5-MST 
, to boost their power \citep{friedman1979,chen2017new}. However, \cite{chen2018weighted} found that, for $k$-MST ($k>1$), the edge-count test ($R_{0,G}$) behaves weirdly when the two sample sizes are different. For example, consider the testing problem that the two underlying distributions are $\mathcal{N}_d(0,\I_d)$ vs $\mathcal{N}_d(\mu,\I_d)$ (e.g., $\|\mu\|_2=1.3,~d=50$), and two scenarios (\textrm{i}) $n_1=n_2=50$ and (\textrm{ii}) $n_1=50,~n_2=100$. The edge-count test has lower power in (\textrm{ii}) compared to that in (\textrm{i}) even though there are more observations in (\textrm{ii}). This is due to a variance boosting issue under unbalanced sample sizes (details seen in \cite{chen2018weighted}). To solve this issue, \cite{chen2018weighted} proposed a weighted edge-count test by inversely weighting the within-sample edges by the sample sizes\footnote{\cite{chen2018weighted} also studied $\frac{n_2}{n_1+n_2}R_{1,G}+\frac{n_1}{n_1+n_2}R_{2,G}$. These two statistics behave very similarly and only differ slightly when the sample sizes are small.}
 \begin{equation}\label{eq:Rw}
 R_{w,G} = \frac{n_2-1}{n_1+n_2-2}R_{1,G} + \frac{n_1-1}{n_1+n_2-2}R_{2,G} 
 \end{equation}
 with the reasoning that the sample with a larger number of observations is more likely to be connected within the sample if all other conditions are the same and thus shall be down-weighted. This weighted edge-count test statistic successfully addressed the variance boosting issue and works well for unequal sample sizes. Indeed, $\var(R_{w,G})\leq\var((1-p)R_{1,G}+pR_{2,G})$ for any $p\in[0,1]$.
 
\subsubsection{Extended basic quantities in the graph-based framework} \label{sec:extended_basic}
In \cite{chen2013}, the authors considered two ways to summarize the test statistics for $R_{0,G}$: averaging ($R_{0,(a)} = \frac{1}{|\mathcal{G}_{C_0}|}\sum_{G\in \mathcal{G}_{C_0}} R_{0,G}$ where $|\mathcal{G}_{C_0}|$ is the number of graphs in $\mathcal{G}_{C_0}$), and union ($R_{0,(u)} = R_{0, \bar{G}_{C_0}}$ where $\bar{G}_{C_0} = \cup\{G\in \mathcal{G}_{C_0}\}$, i.e., if observations $u$ and $v$ are connected in at least one of the graphs in $\mathcal{G}_{C_0}$, then these two observations are connected in $\bar{G}_{C_0}$\footnote{In the following, we somtimes use $\bar{G}$ instead of $\bar{G}_{C_0}$ when there is no confusion for simplicity.}).  Since it is easy to have a lot of graphs in $\mathcal{G}_{C_0}$, it is many times not feasible to compute these two quantities directly.  \cite{chen2013} derived analytic expressions for computing these two quantities in terms of the summary quantities in Table \ref{tab:notation} and $C_0$:
\begin{align*}
R_{0,(a)} & =\sum_{k=1}^{K}\frac{2n_{1k}n_{2k}}{m_k}+\sum_{(u,v)\in C_0}\frac{n_{1u}n_{2v}+n_{1v}n_{2u}}{m_um_v}, \\
R_{0,(u)} & =\sum_{k=1}^{K}n_{1k}n_{2k}+\sum_{(u,v)\in C_0}(n_{1u}n_{2v}+n_{1v}n_{2u}).	
\end{align*}

Similarly, we could defined $R_{1,(a)}$, $R_{1,(u)}$, $R_{2,(a)}$ and $R_{2,(u)}$, and their analytic expressions in terms of the summary quantities in Table \ref{tab:notation} and $C_0$ are given in Lemma \ref{lemma:basic} 

\begin{lemma}\label{lemma:basic}
The analytic expressions for $R_{1,(a)}$, $R_{1,(u)}$, $R_{2,(a)}$ and $R_{2,(u)}$ are:
\begin{align*}
  R_{1,(a)}  & \equiv \frac{1}{|\mathcal{G}_{C_0}|}\sum_{G\in \mathcal{G}_{C_0}} R_{1,G} & = \sum_{u=1}^{K}\frac{n_{1u}(n_{1u}-1)}{m_u}+\sum_{(u,v)\in C_0}\frac{n_{1u}n_{1v}}{m_um_v}, \\
  R_{1,(u)} & \equiv R_{1, \bar{G}_{C_0}} & = \sum_{u=1}^{K}\frac{n_{1u}(n_{1u}-1)}{2}+\sum_{(u,v)\in C_0}n_{1u}n_{1v}, \\
  R_{2,(a)} & \equiv \frac{1}{|\mathcal{G}_{C_0}|}\sum_{G\in \mathcal{G}_{C_0}} R_{2,G} & = \sum_{u=1}^{K}\frac{n_{2u}(n_{2u}-1)}{m_u}+\sum_{(u,v)\in C_0}\frac{n_{2u}n_{2v}}{m_um_v}, \\
  R_{2,(u)} & \equiv R_{2, \bar{G}_{C_0}} & = \sum_{u=1}^{K}\frac{n_{2u}(n_{2u}-1)}{2}+\sum_{(u,v)\in C_0}n_{2u}n_{2v}.
\end{align*}
\end{lemma}

These analytic expressions can be obtained through similar arguments in \cite{chen2013} and the proof is omitted here.


\section{Extended graph-based tests}
\label{sec:extended}

Since the generalized edge-count test could cover a wider range of alternatives than the original edge-count test \citep{chen2017new}, we would like to have the generalized edge-count test statistic well defined when there are repeated observations.  
For the generalized edge-count test statistic: $$S_G=\left( \begin{array}{c}
	R_{1,G} - \E(R_{1,G}) \\ R_{2,G} - \E(R_{2,G}) 
\end{array} \right)^T \mathbf\Sigma_G^{-1}\left( \begin{array}{c}
	R_{1,G} - \E(R_{1,G}) \\ R_{2,G} - \E(R_{2,G}) 
\end{array} \right),$$ one straightforward way of defining the average statistic would be $$\frac{1}{|\mathcal{G}_{C_0}|}\sum_{G\in \mathcal{G}_{C_0}} S_{G}.$$  However, $\mathbf\Sigma_G$ varies for different $G$'s in $\mathcal{G}_{C_0}$, making the averaging over $S_G$'s difficult to move forward.  Even consider the simplified version that $\Sigma_G$ is fixed over $G$'s in $\mathcal{G}_{C_0}$, the quadratic terms in $S_G$ also make the averaging analytically intractable.  To view the problem more straightforwardly, notice that $S_G$ can be written as
\begin{align*}
S_G = \left(\frac{R_{w,G}-\E(R_{w,G})}{\sqrt{\var(R_{w,G})}}\right)^2 + \left(\frac{R_{d,G}-\E(R_{d,G})}{\sqrt{\var(R_{d,G})}}  \right)^2,	
\end{align*}
where $R_{w,G}=\frac{n_2-1}{N-2}R_{1,G}+\frac{n_1-1}{N-2}R_{2,G}$, $R_{d,G}=R_{1,G}-R_{2,G}$, and the two components are asymptotically independent under mild conditions \citep{chu2019asymptotic}.  Let $\E_{\mathcal{G}_{C_0}}$ and $\var_{\mathcal{G}_{C_0}}$ be the expectation and variance defined on the sample space $\mathcal{G}_{C_0}$ that places probability $1/|\mathcal{G}_{C_0}|$ on each $G\in \mathcal{G}_{C_0}$.  Using the first component as an example; taking the averaging over all $G\in \mathcal{G}_{C_0}$ is essentially $\E_{\mathcal{G}_{C_0}}\left( \left(\frac{R_{w,G}-\E(R_{w,G})}{\sqrt{\var(R_{w,G})}}\right)^2 \right) = \left(\E_{\mathcal{G}_{C_0}}\left(\frac{R_{w,G}-\E(R_{w,G})}{\sqrt{\var(R_{w,G})}}\right) \right)^2 + \var_{\mathcal{G}_{C_0}}\left(\frac{R_{w,G}-\E(R_{w,G})}{\sqrt{\var(R_{w,G})}}\right)$.  Here, $${\var(R_{w,G})}=\tfrac{n_1n_2(n_1-1)(n_2-1)}{N(N-1)(N-2)(N-3)}\left(|G|-\tfrac{\sum_{i=1}^N |\mathcal{E}_i^G|^2}{N-2} + \tfrac{2|G|^2}{(N-1)(N-2)}  \right)$$ includes $\sum_{i=1}^N |\mathcal{E}_i^G|^2$, which varies across different $G$'s in $\mathcal{G}_{C_0}$.  So it is already difficult to derive analytic tractable expression even only for $\E_{\mathcal{G}_{C_0}}\left(\frac{R_{w,G}-\E(R_{w,G})}{\sqrt{\var(R_{w,G})}}\right)$.  
To get around the issues, we extend the generalized and weighted edge-count test based on how they were introduced in \cite{chen2018weighted} and \cite{chen2017new}, respectively, based on the extended quantities we have already derived in Section \ref{sec:extended_basic}.  In the following, we first discuss the extended weighted edge-count test, and then the extended generalized edge-count test.  The key components in the extended generalized edge-count test further compose the extended max-type edge-count test. 

%
%

\subsection{Extended weighted edge-count tests}
\subsubsection{Motivation}
\label{moti-rw}

As mentioned in Section \ref{sec:basic}, for data without repeated observations, there is a variance boosting problem for the edge-count test under unbalanced sample sizes. To solve the issue, \cite{chen2018weighted} proposed a weighted edge-count test $R_{w,G}$ (see definition in (\ref{eq:Rw})).

When there are repeated observations, the above problem also exists for the extended edge-count test. To illustrate the problem, we use a preference ranking set up, where two groups of people are asked to rank six objects, and we test whether the two samples have the same preference over these six objects or not. Let $\Xi$ be the set of all permutations of the set $\{1,2,3,4,5,6\}.$ We use the following probability model introduced by \cite{mallows1957} to generate data:
\[
\pp_{\theta,\eta}(\zeta)=\frac{1}{\psi(\theta)}\exp\{-\theta d(\zeta,\eta)\},\quad \zeta,\eta\in\Xi,\quad\theta\in\mathbf{R},
\]
where $d(\cdot,\cdot)$ is a distance function such as Kendall's or Spearman's distance and $\psi$ is a normalizing constant. There are two parameters, $\theta$ and $\eta$, where $\eta$ can be viewed as the ``center'' of the distribution and $\theta$ controls the ``spread'' of the distribution --- the larger $\theta$ is, the less the distribution spreads. In the following, we let $d(\zeta,\eta)$ be the Spearman's distance between $\zeta$ and $\eta$ and let $C_0$ be the 3-NNL on distinct values.

Let $\theta_1=\theta_2=5,$ $\eta_1=\{1,2,3,4,5,6\}$ and $\eta_2=\{1,2,5,4,3,6\}$ in the example. We check the performance under unbalanced sample sizes. The power of $R_{0,(a)}$ and $R_{0,(u)}$ are 0.804 and 0.832 respectively when $n_1=n_2=80$. However, if we increase the sample size of Sample 2 to $n_2=400$ and keep all other parameters unchanged, the power of $R_{0,(a)}$ and $R_{0,(u)}$ decreases to 0.49 and 0.815, respectively (Table \ref{p1.2}).
 \begin{table}[!htbp]
  \centering
  \caption{The fraction of trials (out of 1000) that the test rejected the null hypothesis at 0.05 significance level in the preference ranking example. Here, $\eta_1=\{1,2,3,4,5,6\}$, $\eta_2=\{1,2,5,4,3,6\}$, $\theta_1=\theta_2=5$.}\label{p1.2}
  \begin{tabular}{|c|c|c|}
    \hline
    Power & $n_1=n_2=80$ & $n_1=80,~n_2=400$ \\ \hline
    $R_{0,(a)}$ & 0.804 & 0.49 \\ \hline
    $R_{0,(u)}$ & 0.832 & 0.815 \\
    \hline
  \end{tabular}
\end{table}

\subsubsection{Determining the weights}\label{sec:weights}
Following the similar idea, we could weight $R_{1,(a)}$ and $R_{2,(a)}$, and $R_{1,(u)}$ and $R_{2,(u)}$ to solve the problem. Under the union approach, the statistics $R_{1,(u)}$ and $R_{2,(u)}$ are simplified versions of $R_1$ and $R_2$ defined on $\bar{G}$, so the weights should be the same, i.e., 
\begin{equation}\label{Ruw}
R_{w,(u)} = (1-\hat p)R_{1,(u)} + \hat pR_{2,(u)} \text{ with } \hat p=\frac{n_1-1}{N-2}.
\end{equation}

However, for the average approach, the weights are not this straightforward. The following theorem shows that the weights for the average approach should also be the same. 

\begin{theorem}\label{th:Rwa}
For all test statistics of the form $aR_{1,(a)}+bR_{2,(a)}$, $a+b=1$, $a,b>0$, we have $\var(aR_{1,(a)}+bR_{2,(a)})\geq\var(R_{w,(a)})$, where $R_{w,(a)}=(1-\hat p)R_{1,(a)}+\hat pR_{2,(a)}$ with $\hat p=\frac{n_1-1}{N-2}$.
\end{theorem}

\begin{proof}
It is not hard to see that the minimum is achieved at
\begin{equation}\label{eq_q}
  \hat{p}=\frac{\var(R_{1,(a)})-\cov(R_{1,(a)},R_{2,(a)})}{\var(R_{1,(a)})+\var(R_{2,(a)})-2\cov(R_{1,(a)},R_{2,(a)})}.
\end{equation}

Plugging $\var(R_{1,(a)}),~\var(R_{2,(a)})$ and $\cov(R_{1,(a)},R_{2,(a)})$ provided in Lemma \ref{lem1} into (\ref{eq_q}), we have $\hat{p}=\frac{n_1-1}{N-2}$.
\end{proof}

In the following lemma, we provide exact analytic formulas to the expectation and variance of $R_{w,(u)}$ and $R_{w,(a)}$, respectively, so that both extended weighted edge-count tests can be standardized easily.
\begin{lemma}\label{lemma:Rw}
  The expectation and variance of $R_{w,(u)}$ and $R_{w,(a)}$ under the permutation null are:
\begin{flalign*}
\begin{split}
 \E(R_{w,(u)}) & = |\bar{G}|\tfrac{(n_1-1)(n_2-1)}{(N-1)(N-2)},\\
 \var(R_{w,(u)})& = \tfrac{n_1(n_1-1)n_2(n_2-1)}{N(N-1)(N-2)(N-3)}\Bigg\{|\bar{G}|-\tfrac{1}{N-2}\sum_{i=1}^N|\e_i^{\bar{G}}|^2+\tfrac{2}{(N-1)(N-2)}|\bar{G}|^2\Bigg\},\\
\E(R_{w,(a)}) &=  (N-K+|C_0|)\tfrac{(n_1-1)(n_2-1)}{(N-1)(N-2)},\\
\var(R_{w,(a)}) &= \tfrac{n_1(n_1-1)n_2(n_2-1)}{N(N-1)(N-2)(N-3)}\Bigg\{-\tfrac{4}{N-2}\left(\sum_u\tfrac{(|\e_u^{C_0}|-2)^2}{4m_u}-\tfrac{(|C_0|-K)^2}{N}\right)\\
 & \hspace{2mm}  +2(K-\sum_{u}\tfrac{1}{m_u})+\sum_{(u,v)\in C_0}\tfrac{1}{m_um_v}-\tfrac{2}{N(N-1)}(|C_0|+N-K)^2 \Bigg\},
\end{split}&
\end{flalign*}
where $|\e_i^{\bar{G}}|=m_u-1+\sum_{v\in \V_u^{C_0}}m_v$ if observation $i$ is of value indexed by $u$, and $|\bar{G}|=\sum_{u=1}^{K}m_u(m_u-1)/2+\sum_{(u,v)\in C_0}m_um_v$.  Here, $\V_u^{C_0}$ is the set of distinct values that connect to the distinct value indexed by $u$ in $C_0$.
\end{lemma}

This lemma can be proved straightforwardly by plugging the analytic expressions of $\E(R_{1,(a)})$, $\E(R_{2,(a)})$, $\var(R_{1,(a)})$, $\var(R_{2,(a)})$, $\cov(R_{1,(a)},R_{2,(a)})$, $\E(R_{1,(u)})$, $\E(R_{2,(u)})$, $\var(R_{1,(u)})$, $\var(R_{2,(u)})$ and $\cov(R_{1,(u)},R_{2,(u)})$ provided in Lemmas \ref{lem1} and \ref{lem2} in Appendix \ref{sec:AppA}.


\subsection{Extended generalized edge-count tests}
\label{subsec:general}

As we discussed earlier, it is technically intractable to derive the analytic expression for the average of $S_G$'s for $G\in\mathcal{G}_{C_0}$.  Here, we define extended generalized edge-count test statistic based on how the statistic was introduced in \cite{chen2017new} through the extended basic quantities:
\begin{align}
S_{(a)} & = \left( \begin{array}{c}
	R_{1,(a)} - \E(R_{1,(a)}) \\ R_{2,(a)} - \E(R_{2,(a)}) 
\end{array} \right)^T \mathbf\Sigma_{(a)}^{-1}\left( \begin{array}{c}
	R_{1,(a)} - \E(R_{1,(a)}) \\ R_{2,(a)} - \E(R_{2,(a)})
\end{array} \right),  \label{eq:Sa} \\
S_{(u)} & = \left( \begin{array}{c}
	R_{1,(u)} - \E(R_{1,(u)}) \\ R_{2,(u)} - \E(R_{2,(u)}) 
\end{array} \right)^T \mathbf\Sigma_{(u)}^{-1}\left( \begin{array}{c}
	R_{1,(u)} - \E(R_{1,(u)}) \\ R_{2,(u)} - \E(R_{2,(u)})
\end{array} \right), \label{eq:Su}
\end{align}
where $\Sigma_{(a)}= \var(\binom{R_{1,(a)}}{R_{2,(a)}})$, $\Sigma_{(u)}= \var(\binom{R_{1,(u)}}{R_{2,(u)}})$.  With similar arguments in \cite{chen2017new}, $S_{(a)}$ and $S_{(u)}$ defined in this way could deal with location and scale alternatives.  More studies on the performance of the tests are in Section \ref{sec:performance}.
Similar to $S_G$, $S_{(a)}$ and $S_{(u)}$ defined above can also be decomposed to components that are asymptotically independent under mild conditions, respectively (details see Theorems \ref{theorem:gauss1} and \ref{theorem:gauss2}). 

\begin{lemma}\label{lemma:SaSu}
	The extended generalized edge-count test statistics can be expressed as
	\begin{align}
		S_{(a)} & = \left(\frac{R_{w,(a)}-\E(R_{w,(a)})}{\sqrt{\var(R_{w,(a)})}}\right)^2 + \left(\frac{R_{d,(a)}-\E(R_{d,(a)})}{\sqrt{\var(R_{d,(a)})}}  \right)^2, \label{eq:Sa2}\\
		S_{(u)} & = \left(\frac{R_{w,(u)}-\E(R_{w,(u)})}{\sqrt{\var(R_{w,(u)})}}\right)^2 + \left(\frac{R_{d,(u)}-\E(R_{d,(u)})}{\sqrt{\var(R_{d,(u)})}}  \right)^2,	\label{eq:Su2}
	\end{align}
where $R_{w,(a)}$, $\E(R_{w,(a)})$, $\var(R_{w,(a)})$, $R_{w,(u)}$, $\E(R_{w,(u)})$ and $\var(R_{w,(u)})$ are provided in Section \ref{sec:weights}, and $R_{d,(a)} = R_{1,(a)} - R_{2,(a)}$, $R_{d,(u)} = R_{1,(u)} - R_{2,(u)}$ with their expectations and variances provided below.  
\begin{flalign*}
\begin{split}
 \E(R_{d,(a)})& =(N-K+|C_0|)\tfrac{n_1-n_2}{N},\\
 \var(R_{d,(a)})& = \tfrac{4n_1n_2}{N(N-1)}\Bigg\{\sum_u\tfrac{(|\e_u^{C_0}|-2)^2}{4m_u}-\tfrac{(|C_0|-K)^2}{N}\Bigg\}, \\
 \E(R_{d,(u)})&=  |\bar{G}|\tfrac{n_1-n_2}{N},\\
 \var(R_{d,(u)})&= \tfrac{n_1n_2}{N(N-1)}\Bigg\{\sum_{i=1}^N|\e_i^{\bar{G}}|^2-\tfrac{4}{N}|\bar{G}|^2\Bigg\}.
\end{split}&
\end{flalign*}

\end{lemma}

Lemma \ref{lemma:SaSu} is proved in supplementary materials.

\subsection{Extended max-type edge-count test statistics}
Let $Z_{w,(a)} = \frac{R_{w,(a)}-\E(R_{w,(a)})}{\sqrt{\var(R_{w,(a)})}}$, $Z_{d,(a)} = \frac{R_{d,(a)}-\E(R_{d,(a)})}{\sqrt{\var(R_{d,(a)})}}$, $Z_{w,(u)} = \frac{R_{w,(u)}-\E(R_{w,(u)})}{\sqrt{\var(R_{w,(u)})}}$, and $Z_{d,(u)}= \frac{R_{d,(u)}-\E(R_{d,(u)})}{\sqrt{\var(R_{d,(u)})}}$.
Under some mild conditions, $Z_{w,(a)}$ and $Z_{d,(a)}$ are asymptotically independent with their joint distribution bivariate normal, and same for $Z_{w,(u)}$ and $Z_{d,(u)}$ (details see Theorems \ref{theorem:gauss1} and \ref{theorem:gauss2}). Here, we define the extended max-type edge-count statistics: 
\begin{equation*}
  M_{(a)}(\kappa) = \max(\kappa Z_{w,(a)},|Z_{d,(a)}|),\text{ and } M_{(u)}(\kappa) = \max(\kappa Z_{w,(u)},|Z_{d,(u)}|).
\end{equation*}

As the following arguments hold the same for the averaging the union statistics, we omit subscripts $(a)$ and $(u)$ for simplicity.  From the definition of the extended max-type edge-count test statistic, we can see that it makes use of both $Z_w$ and $Z_d$, and would be similar to $S$ and effective to both location and scale alternatives.  Also, the introduction of $\kappa$ in the definition makes it more flexible than $S$. 


We next briefly discuss the choice of $\kappa$. 
 It is easy to see that the rejection region $\{M(\kappa)\geq\beta\}$ is equivalent to $\{Z_{w}\geq\frac{\beta}{\kappa}\text{ or }|Z_{d}|\geq\beta\}$.
Let $\PP(Z_w\geq\beta_w)=\alpha_1$ and $\PP(|Z_d|\geq\beta_d)=\alpha_2$, and define $\gamma=\frac{\alpha_1}{\alpha_2}$.  Based on the asymptotic distribution of $(Z_w, Z_d)^T$ derived in Section \ref{sec:asym}, the relationship between $\gamma$ and $\kappa$ with the overall type I error rate controlled at 0.05 is shown in Table \ref{T2}.

 \begin{table}[!ht]
  \centering
  \caption{Relationship between $\gamma$ and $\kappa$.}\label{T2}
  \begin{tabular}{|cccccccc|}
    \hline
    $\gamma$ & 8 & 4 & 2 & 1 & 1/2 & 1/4 & 1/8\\ \hline
    $\kappa$ & 1.63 & 1.47 & 1.31 & 1.14 & 1 & 0.88 & 0.79 \\ \hline
  \end{tabular}
  \end{table}

To check how the choice of $\kappa$ affects the performance of the test, we examine the test on 100-dimensional multivariate normal distributions $\mathcal{N}_d(\mu_1,\Sigma_1)$ and $\mathcal{N}_d(\mu_2,\Sigma_2)$ that are different in mean and/or variance. 
Three scenarios are considered and the detailed results are presented in supplementary materials. Based on the simulation results, if there is no prior knowledge about the type of difference between the two distributions, we recommend $\kappa=\{1.31,1.14,1\}$ for $M(\kappa)$.

\section{Performance of the extended test statistics}
\label{sec:performance}
In this section, we study the performance of various tests through the ranking problems, where two groups of people are asked to rank six objects, and we test whether the two samples have the same preference over these six objects or not.  We consider the following two data generating mechanisms.  

\begin{enumerate}[(i)]
\item Data are genearated from the probability model introduced in Section \ref{moti-rw}
\begin{equation}\label{rankmodel}
\pp_{\theta,\eta}(\zeta)=\frac{1}{\psi(\theta)}\exp\{-\theta d(\zeta,\eta)\},\quad \zeta,\eta\in\Xi,\quad\theta\in\mathbf{R},
\end{equation}
where $\Xi$ be the set of all permutations of the set \{1,2,3,4,5,6\} and $d(\cdot,\cdot)$ is a distance function such as Kendall's or Spearman's distance.  The two samples are generated from $\pp_{\theta_1,\eta_1}(\cdot)$ and $\pp_{\theta_2,\eta_2}(\cdot)$, respectively.
\item  Let $\CalD_1$ and $\CalD_2$ be two different subsets of all possible rankings.  The two sample are generated from the uniform distribution on $\CalD_1$ and $\CalD_2$, respectively. 
\end{enumerate}

When Kendall's or Spearman's distance is used for $d(\cdot,\cdot)$, there are in general ties in the distance matrix, which lead to non-unique MSTs. Hence, we apply 3-NNL to construct the graph on distinct values. The results for Kendall's and Spearman's distance are very similar,  so we present the results based on the Spearman's distance in the following.

We compare the following statistics: $R_{0,(a)},~R_{0,(u)},$ $S_{(a)}$, $S_{(u)}$, $R_{w,(a)}$, $R_{w,(u)}$, $M_{(a)}(\kappa)$ and $M_{(u)}(\kappa)$ ($\kappa=1.31, 1.14, 1$) in eight scenarios (Scenarios 1--5 under (i) and Scenarios 6--8 under (ii)) with balanced and unbalanced sample sizes. In each scenario, the specific parameters under each scenario are chosen such that the tests have moderate power to be comparable.
\vspace{-0.2cm}
 \begin{itemize}
  \item Scenario 1 (Only $\eta$ differs) :

  $\eta_1=\{1,2,3,4,5,6\}$, $\eta_2=\{1,2,5,4,3,6\}$, $\theta_1=\theta_2=5$ with balanced ($n_1=n_2=100$) and unbalance ($n_1=100, n_2=400$) sample sizes.
  \item Scenario 2 (Only $\theta$ differs with $\theta_1>\theta_2$) :

  $\eta_1=\eta_2=\{1,2,3,4,5,6\}$, $\theta_1=5.5,~\theta_2=4$ with balanced ($n_1=n_2=300$) and unbalance ($n_1=300, n_2=600$) sample sizes.
  \item Scenario 3 (Only $\theta$ differs with $\theta_1<\theta_2$) :

  $\eta_1=\eta_2=\{1,2,3,4,5,6\}$, $\theta_1=4,~\theta_2=5.5$ with balanced ($n_1=n_2=300$) and unbalance ($n_1=300, n_2=600$) sample sizes.
  \item Scenario 4 (Both $\eta$ and $\theta$ differ with $\theta_1>\theta_2$) :

  $\eta_1=\{1,2,3,4,5,6\}$, $\eta_2=\{1,2,5,4,3,6\}$, $\theta_1=5.5,~\theta_2=4$ with balanced ($n_1=n_2=100$) and unbalance ($n_1=100, n_2=300$) sample sizes.
  \item Scenario 5 (Both $\eta$ and $\theta$ differ with $\theta_1<\theta_2$) :

  $\eta_1=\{1,2,3,4,5,6\}$, $\eta_2=\{1,2,5,4,3,6\}$, $\theta_1=4,~\theta_2=5.5$ with balanced ($n_1=n_2=100$) and unbalance ($n_1=100, n_2=300$) sample sizes.
  \item Scenario 6 (Different supports):
  
  $\CalD_1= \{\zeta\in\Xi:\zeta \text{ does not begin with No.6}\}$, $\CalD_2=\{\zeta \in\Xi:\zeta \text{ does not}$ $\text{end with No.1}\}$ with balanced ($n_1=n_2=150$) and unbalance ($n_1=150, n_2=250$) sample sizes.
  \item Scenario 7 (Different supports):
  
  $\CalD_1= \{\zeta\in\Xi:\zeta \text{ ranks No.1 before No.5}\}$, $\CalD_2=\{\zeta \in\Xi:\zeta \text{ ranks No.1}$ $\text{before No.6}\}$ with balanced ($n_1=n_2=150$) and unbalance ($n_1=150, n_2=250$) sample sizes.
  \item Scenario 8 (Different supports):
  
  $\CalD_1= \{\zeta\in\Xi:\zeta \text{ does not begin with No.6 and does not end with No.1}\}$, $\CalD_2=\{\zeta \in\Xi:\zeta \text{ ranks No.1 or No.2 in top 3}\}$ with balanced ($n_1=n_2=150$) and unbalance ($n_1=150, n_2=250$) sample sizes.
 \end{itemize}

\vspace{0cm}
\begin{table}[!htp]
  \centering
  \caption{Scenario 1: $\eta_1=\{1,2,3,4,5,6\}$, $\eta_2=\{1,2,5,4,3,6\}$, $\theta_1=\theta_2=5$}\label{Rank1}
  \begin{tabular}{|c|c|c|c|c|c|c|}
    \hline
    \multicolumn{7}{|c|}{$n_1=n_2=100$} \\ \hline
    Statistic & $R_{0,(a)}$ & $S_{(a)}$ & $R_{w,(a)}$ & $M_{(a)}(1.31)$ & $M_{(a)}(1.14)$ & $M_{(a)}(1)$ \\ \hline
    Estimated Power & \textbf{0.857} & 0.750 & \textbf{0.857} & 0.831 & 0.813 & 0.780 \\ \hline
    Statistic & $R_{0,(u)}$ & $S_{(u)}$ & $R_{w,(u)}$ & $M_{(u)}(1.31)$ & $M_{(u)}(1.14)$ & $M_{(u)}(1)$ \\ \hline
    Estimated Power & \textbf{0.888} & 0.791 & \textbf{0.888} & \textbf{0.861} & 0.840 & 0.818 \\
    \hline
  \end{tabular}
  \begin{tabular}{|c|c|c|c|c|c|c|}
    \hline
    \multicolumn{7}{|c|}{$n_1=100, n_2=400$} \\ \hline
    Statistic & $R_{0,(a)}$ & $S_{(a)}$ & $R_{w,(a)}$ & $M_{(a)}(1.31)$ & $M_{(a)}(1.14)$ & $M_{(a)}(1)$ \\ \hline
    Estimated Power & 0.641 & 0.889 & \textbf{0.949} & \textbf{0.940} & \textbf{0.935} & 0.915 \\ \hline
    Statistic & $R_{0,(u)}$ & $S_{(u)}$ & $R_{w,(u)}$ & $M_{(u)}(1.31)$ & $M_{(u)}(1.14)$ & $M_{(u)}(1)$ \\ \hline
    Estimated Power & 0.871 & \textbf{0.951} & \textbf{0.977} & \textbf{0.969} & \textbf{0.961} & \textbf{0.959} \\
    \hline
  \end{tabular}
\end{table}

\begin{table}[!htp]
  \centering
  \caption{Scenario 2: $\eta_1=\eta_2=\{1,2,3,4,5,6\}$, $\theta_1=5.5,~\theta_2=4$.}\label{Rank2}
  \begin{tabular}{|c|c|c|c|c|c|c|}
    \hline
    \multicolumn{7}{|c|}{$n_1=n_2=300$} \\ \hline
    Statistic & $R_{0,(a)}$ & $S_{(a)}$ & $R_{w,(a)}$ & $M_{(a)}(1.31)$ & $M_{(a)}(1.14)$ & $M_{(a)}(1)$ \\ \hline
    Estimated Power & {0.265} & 0.172 & {0.265} & 0.239 & 0.223 & 0.194 \\ \hline
    Statistic & $R_{0,(u)}$ & $S_{(u)}$ & $R_{w,(u)}$ & $M_{(u)}(1.31)$ & $M_{(u)}(1.14)$ & $M_{(u)}(1)$ \\ \hline
    Estimated Power & 0.438 & \textbf{0.796} & 0.438 & 0.767 & \textbf{0.797} & \textbf{0.828} \\
    \hline
  \end{tabular}
  \begin{tabular}{|c|c|c|c|c|c|c|}
    \hline
    \multicolumn{7}{|c|}{$n_1=300, n_2=600$} \\ \hline
    Statistic & $R_{0,(a)}$ & $S_{(a)}$ & $R_{w,(a)}$ & $M_{(a)}(1.31)$ & $M_{(a)}(1.14)$ & $M_{(a)}(1)$ \\ \hline
    Estimated Power & {0.525} & 0.325 & 0.310 & 0.348 & 0.334 & 0.318 \\ \hline
    Statistic & $R_{0,(u)}$ & $S_{(u)}$ & $R_{w,(u)}$ & $M_{(u)}(1.31)$ & $M_{(u)}(1.14)$ & $M_{(u)}(1)$ \\ \hline
    Estimated Power & 0 & \textbf{0.899} & 0.566 & \textbf{0.887} & \textbf{0.912} & \textbf{0.929} \\
    \hline
  \end{tabular}
\end{table}

\begin{table}[!htp]
  \centering
  \caption{Scenario 3: $\eta_1=\eta_2=\{1,2,3,4,5,6\}$, $\theta_1=4,~\theta_2=5.5$. }\label{Rank3}
  \begin{tabular}{|c|c|c|c|c|c|c|}
    \hline
    \multicolumn{7}{|c|}{$n_1=n_2=300$} \\ \hline
    Statistic & $R_{0,(a)}$ & $S_{(a)}$ & $R_{w,(a)}$ & $M_{(a)}(1.31)$ & $M_{(a)}(1.14)$ & $M_{(a)}(1)$ \\ \hline
    Estimated Power & {0.279} & 0.181 & {0.279} & 0.250 & 0.231 & 0.208 \\ \hline
    Statistic & $R_{0,(u)}$ & $S_{(u)}$ & $R_{w,(u)}$ & $M_{(u)}(1.31)$ & $M_{(u)}(1.14)$ & $M_{(u)}(1)$ \\ \hline
    Estimated Power & 0.413 & 0.755 & 0.413 & 0.730 & \textbf{0.781} & \textbf{0.806} \\
    \hline
  \end{tabular}
  \begin{tabular}{|c|c|c|c|c|c|c|}
    \hline
    \multicolumn{7}{|c|}{$n_1=300, n_2=600$} \\ \hline
    Statistic & $R_{0,(a)}$ & $S_{(a)}$ & $R_{w,(a)}$ & $M_{(a)}(1.31)$ & $M_{(a)}(1.14)$ & $M_{(a)}(1)$ \\ \hline
    Estimated Power & 0.061 & 0.378 & 0.355 & {0.393} & {0.393} & 0.386 \\ \hline
    Statistic & $R_{0,(u)}$ & $S_{(u)}$ & $R_{w,(u)}$ & $M_{(u)}(1.31)$ & $M_{(u)}(1.14)$ & $M_{(u)}(1)$ \\ \hline
    Estimated Power & \textbf{0.954} & 0.899 & 0.545 & 0.874 & \textbf{0.909} & \textbf{0.922} \\
    \hline
  \end{tabular}
\end{table}

\begin{table}[!htp]
  \centering
  \caption{Scenario 4: $\eta_1=\{1,2,3,4,5,6\}$, $\eta_2=\{1,2,5,4,3,6\}$, $\theta_1=5.5,~\theta_2=4$. The largest number in each row is boldened.}\label{Rank4}
  \begin{tabular}{|c|c|c|c|c|c|c|}
    \hline
    \multicolumn{7}{|c|}{$n_1=n_2=100$} \\ \hline
    Statistic & $R_{0,(a)}$ & $S_{(a)}$ & $R_{w,(a)}$ & $M_{(a)}(1.31)$ & $M_{(a)}(1.14)$ & $M_{(a)}(1)$ \\ \hline
    Estimated Power & \textbf{0.848} & 0.754 & \textbf{0.848} & 0.821 & 0.805 & 0.778 \\ \hline
    Statistic & $R_{0,(u)}$ & $S_{(u)}$ & $R_{w,(u)}$ & $M_{(u)}(1.31)$ & $M_{(u)}(1.14)$ & $M_{(u)}(1)$ \\ \hline
    Estimated Power & \textbf{0.884} & \textbf{0.865} & \textbf{0.884} & \textbf{0.883} & \textbf{0.879} & \textbf{0.863} \\
    \hline
  \end{tabular}
  \begin{tabular}{|c|c|c|c|c|c|c|}
    \hline
    \multicolumn{7}{|c|}{$n_1=100, n_2=300$} \\ \hline
    Statistic & $R_{0,(a)}$ & $S_{(a)}$ & $R_{w,(a)}$ & $M_{(a)}(1.31)$ & $M_{(a)}(1.14)$ & $M_{(a)}(1)$ \\ \hline
    Estimated Power & 0.790 & 0.888 & \textbf{0.948} & \textbf{0.940} & \textbf{0.925} & 0.912 \\ \hline
    Statistic & $R_{0,(u)}$ & $S_{(u)}$ & $R_{w,(u)}$ & $M_{(u)}(1.31)$ & $M_{(u)}(1.14)$ & $M_{(u)}(1)$ \\ \hline
    Estimated Power & 0.493 & \textbf{0.952} & \textbf{0.970} & \textbf{0.965} & \textbf{0.965} & \textbf{0.954} \\
    \hline
  \end{tabular}
\end{table}

\begin{table}[!htp]
  \centering
  \caption{Scenario 5: $\eta_1=\{1,2,3,4,5,6\}$, $\eta_2=\{1,2,5,4,3,6\}$, $\theta_1=4,~\theta_2=5.5$. }\label{Rank5}
  \begin{tabular}{|c|c|c|c|c|c|c|}
    \hline
    \multicolumn{7}{|c|}{$n_1=n_2=100$} \\ \hline
    Statistic & $R_{0,(a)}$ & $S_{(a)}$ & $R_{w,(a)}$ & $M_{(a)}(1.31)$ & $M_{(a)}(1.14)$ & $M_{(a)}(1)$ \\ \hline
    Estimated Power & \textbf{0.888} & 0.778 & \textbf{0.888} & 0.854 & 0.834 & 0.805 \\ \hline
    Statistic & $R_{0,(u)}$ & $S_{(u)}$ & $R_{w,(u)}$ & $M_{(u)}(1.31)$ & $M_{(u)}(1.14)$ & $M_{(u)}(1)$ \\ \hline
    Estimated Power & \textbf{0.917} & \textbf{0.873} & \textbf{0.917} & \textbf{0.898} & \textbf{0.890} & 0.870 \\
    \hline
  \end{tabular}
  \begin{tabular}{|c|c|c|c|c|c|c|}
    \hline
    \multicolumn{7}{|c|}{$n_1=100, n_2=300$} \\ \hline
    Statistic & $R_{0,(a)}$ & $S_{(a)}$ & $R_{w,(a)}$ & $M_{(a)}(1.31)$ & $M_{(a)}(1.14)$ & $M_{(a)}(1)$ \\ \hline
    Estimated Power & 0.813 & 0.917 & \textbf{0.962} & \textbf{0.954} & \textbf{0.947} & 0.935 \\ \hline
    Statistic & $R_{0,(u)}$ & $S_{(u)}$ & $R_{w,(u)}$ & $M_{(u)}(1.31)$ & $M_{(u)}(1.14)$ & $M_{(u)}(1)$ \\ \hline
    Estimated Power & \textbf{0.996} & \textbf{0.993} & \textbf{0.985} & \textbf{0.986} & \textbf{0.986} & \textbf{0.989} \\
    \hline
  \end{tabular}
\end{table}

\begin{table}[!htp]
  \centering
  \caption{Scenario 6: $\CalD_1= \{\zeta\in\Xi:\zeta \text{ does not begin with No.6}\}$, $\CalD_2=\{\zeta \in\Xi:\zeta \text{ does not}$ $\text{end with No.1}\}$}\label{Rank01}
  \begin{tabular}{|c|c|c|c|c|c|c|}
    \hline
    \multicolumn{7}{|c|}{$n_1=n_2=150$} \\ \hline
    Statistic & $R_{0,(a)}$ & $S_{(a)}$ & $R_{w,(a)}$ & $M_{(a)}(1.31)$ & $M_{(a)}(1.14)$ & $M_{(a)}(1)$ \\ \hline
    Estimated Power & \textbf{0.745} & 0.557 & \textbf{0.745} & 0.695 & 0.646 & 0.594 \\ \hline
    Statistic & $R_{0,(u)}$ & $S_{(u)}$ & $R_{w,(u)}$ & $M_{(u)}(1.31)$ & $M_{(u)}(1.14)$ & $M_{(u)}(1)$ \\ \hline
    Estimated Power & 0.670 & 0.503 & 0.670 & 0.626 & 0.580 & 0.528 \\
    \hline
  \end{tabular}
  \begin{tabular}{|c|c|c|c|c|c|c|}
    \hline
    \multicolumn{7}{|c|}{$n_1=150, n_2=250$} \\ \hline
    Statistic & $R_{0,(a)}$ & $S_{(a)}$ & $R_{w,(a)}$ & $M_{(a)}(1.31)$ & $M_{(a)}(1.14)$ & $M_{(a)}(1)$ \\ \hline
    Estimated Power & 0.826 & 0.744 & \textbf{0.881} & 0.834 & 0.804 & 0.767 \\ \hline
    Statistic & $R_{0,(u)}$ & $S_{(u)}$ & $R_{w,(u)}$ & $M_{(u)}(1.31)$ & $M_{(u)}(1.14)$ & $M_{(u)}(1)$ \\ \hline
    Estimated Power &  0.782 & 0.637 & 0.783 & 0.746 & 0.714 & 0.668 \\
    \hline
  \end{tabular}
\end{table}

\begin{table}[!htp]
  \centering
  \caption{Scenario 7: $\CalD_1= \{\zeta\in\Xi:\zeta \text{ ranks No.1 before No.5}\}$, $\CalD_2=\{\zeta \in\Xi:\zeta \text{ ranks No.1}$ $\text{before No.6}\}$}\label{Rank02}
  \begin{tabular}{|c|c|c|c|c|c|c|}
    \hline
    \multicolumn{7}{|c|}{$n_1=n_2=150$} \\ \hline
    Statistic & $R_{0,(a)}$ & $S_{(a)}$ & $R_{w,(a)}$ & $M_{(a)}(1.31)$ & $M_{(a)}(1.14)$ & $M_{(a)}(1)$ \\ \hline
    Estimated Power & \textbf{0.620} & 0.447 & \textbf{0.620} & 0.573 & 0.528 & 0.468 \\ \hline
    Statistic & $R_{0,(u)}$ & $S_{(u)}$ & $R_{w,(u)}$ & $M_{(u)}(1.31)$ & $M_{(u)}(1.14)$ & $M_{(u)}(1)$ \\ \hline
    Estimated Power & 0.502 & 0.387 & 0.502 & 0.470 & 0.450 & 0.415 \\
    \hline
  \end{tabular}
  \begin{tabular}{|c|c|c|c|c|c|c|}
    \hline
    \multicolumn{7}{|c|}{$n_1=150, n_2=250$} \\ \hline
    Statistic & $R_{0,(a)}$ & $S_{(a)}$ & $R_{w,(a)}$ & $M_{(a)}(1.31)$ & $M_{(a)}(1.14)$ & $M_{(a)}(1)$ \\ \hline
    Estimated Power & \textbf{0.840} & 0.743 & \textbf{0.880} & \textbf{0.841} & 0.815 & 0.790 \\ \hline
    Statistic & $R_{0,(u)}$ & $S_{(u)}$ & $R_{w,(u)}$ & $M_{(u)}(1.31)$ & $M_{(u)}(1.14)$ & $M_{(u)}(1)$ \\ \hline
    Estimated Power & 0.834 & 0.661 & 0.698 & 0.692 & 0.683 & 0.647 \\
    \hline
  \end{tabular}
\end{table}

\begin{table}[!htp]
  \centering
  \caption{Scenario 8: $\CalD_1= \{\zeta\in\Xi:\zeta \text{ does not begin with No.6 and does not end with No.1}\}$, $\CalD_2=\{\zeta \in\Xi:\zeta \text{ ranks No.1 or No.2 in top 3}\}$}\label{Rank03}
  \begin{tabular}{|c|c|c|c|c|c|c|}
    \hline
    \multicolumn{7}{|c|}{$n_1=n_2=150$} \\ \hline
    Statistic & $R_{0,(a)}$ & $S_{(a)}$ & $R_{w,(a)}$ & $M_{(a)}(1.31)$ & $M_{(a)}(1.14)$ & $M_{(a)}(1)$ \\ \hline
    Estimated Power & \textbf{0.886} & 0.763 & \textbf{0.886} & \textbf{0.858} & 0.828 & 0.788 \\ \hline
    Statistic & $R_{0,(u)}$ & $S_{(u)}$ & $R_{w,(u)}$ & $M_{(u)}(1.31)$ & $M_{(u)}(1.14)$ & $M_{(u)}(1)$ \\ \hline
    Estimated Power & 0.814 & 0.681 & 0.814 & 0.774 & 0.745 & 0.708\\
    \hline
  \end{tabular}
  \begin{tabular}{|c|c|c|c|c|c|c|}
    \hline
    \multicolumn{7}{|c|}{$n_1=150, n_2=250$} \\ \hline
    Statistic & $R_{0,(a)}$ & $S_{(a)}$ & $R_{w,(a)}$ & $M_{(a)}(1.31)$ & $M_{(a)}(1.14)$ & $M_{(a)}(1)$ \\ \hline
    Estimated Power & \textbf{0.943} & \textbf{0.916} & \textbf{0.962} & \textbf{0.944} & \textbf{0.938} & \textbf{0.928} \\ \hline
    Statistic & $R_{0,(u)}$ & $S_{(u)}$ & $R_{w,(u)}$ & $M_{(u)}(1.31)$ & $M_{(u)}(1.14)$ & $M_{(u)}(1)$ \\ \hline
    Estimated Power & 0.888 & 0.821 & \textbf{0.917} & 0.895 & 0.885 & 0.852 \\
    \hline
  \end{tabular}
\end{table}

The results are presented in Tables \ref{Rank1}--\ref{Rank03}.  Each table lists the fraction of trials (out of 1000) that the test reject the null hypothesis at 0.05 significance level. Those above 95 percentage of the best power under each setting are in bold. 

Tables \ref{Rank1}--\ref{Rank5} provide results for the data generated by mechanism (i). We see that $S_{(u)}$ and $M_{(u)}$ work well for all scenarios, while the others show obvious strengthes and weaknesses for different settings. For example, under the unbalanced setting ($n_1=300, n_2=600$), $R_{0,(u)}$ has no power under Scenario 2, $R_{0,(a)}$ has very low power under Scenario 3, and both $R_{w,(a)}$ and $R_{w,(u)}$ do not perform well when only $\theta$ differs (Scenarios 2 and 3). Overall, $M_{(u)}(\kappa)$ perform best among all the tests. When $\theta$ differs, $S_{(a)}$ and $S_{(u)}$ provide similar results to $M_{(a)}(\kappa)$ and $M_{(u)}(\kappa)$, respectively, but they perform worse than $M_{(a)}(\kappa)$ and $M_{(u)}(\kappa)$, respectively, when only $\eta$ differs (Scenario 1). In general, the tests based on ``union'' are slightly better than their ``averaging'' counterparts (except for some cases for $R_0$). 

 Tables \ref{Rank01}--\ref{Rank03} provide results for data generated by mechanism (ii).  We 
  see that the tests perform similarly well with those based on ``averaging'' slightly better than their ``union'' counterparts.

\begin{remark}
	For either the ``averaging" statistics and the ``union" statistics, their relationships can be represented by the following schematic plots on the reject regions in terms of $Z_w$ and $Z_d$.
	\begin{figure}[!h]
    \centering
    \includegraphics[width=0.32\textwidth]{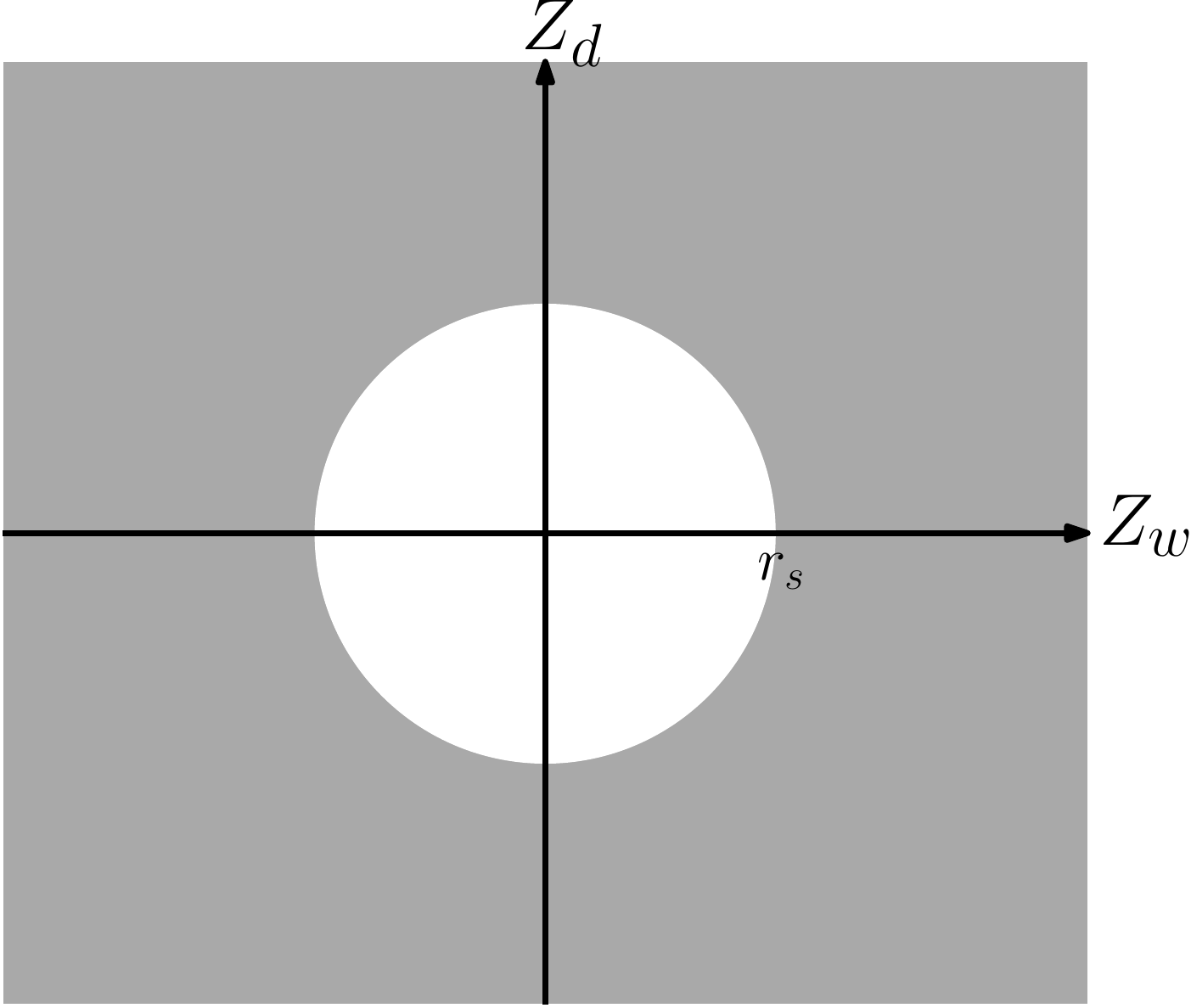}
    \includegraphics[width=0.32\textwidth]{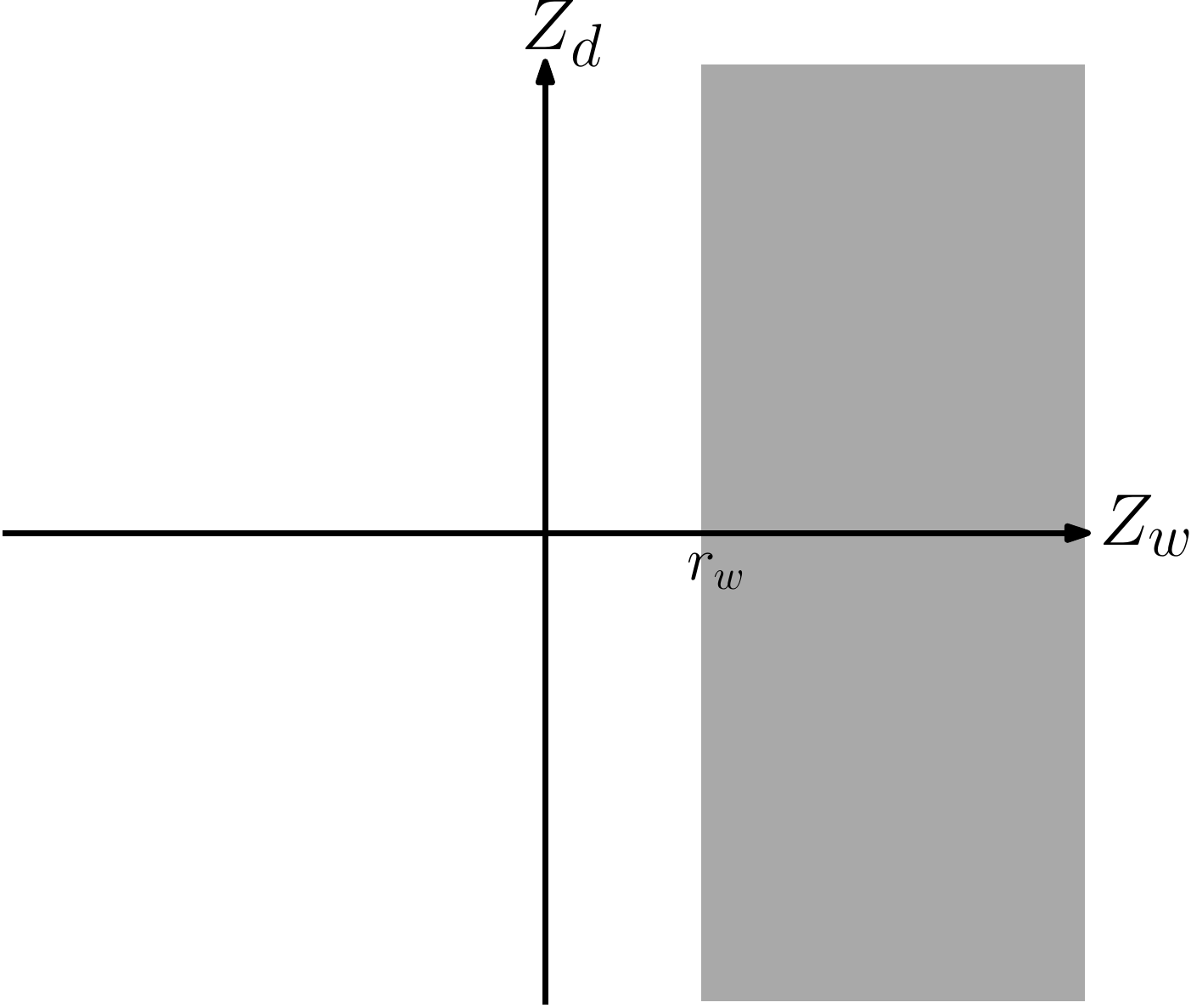}
    \includegraphics[width=0.32\textwidth]{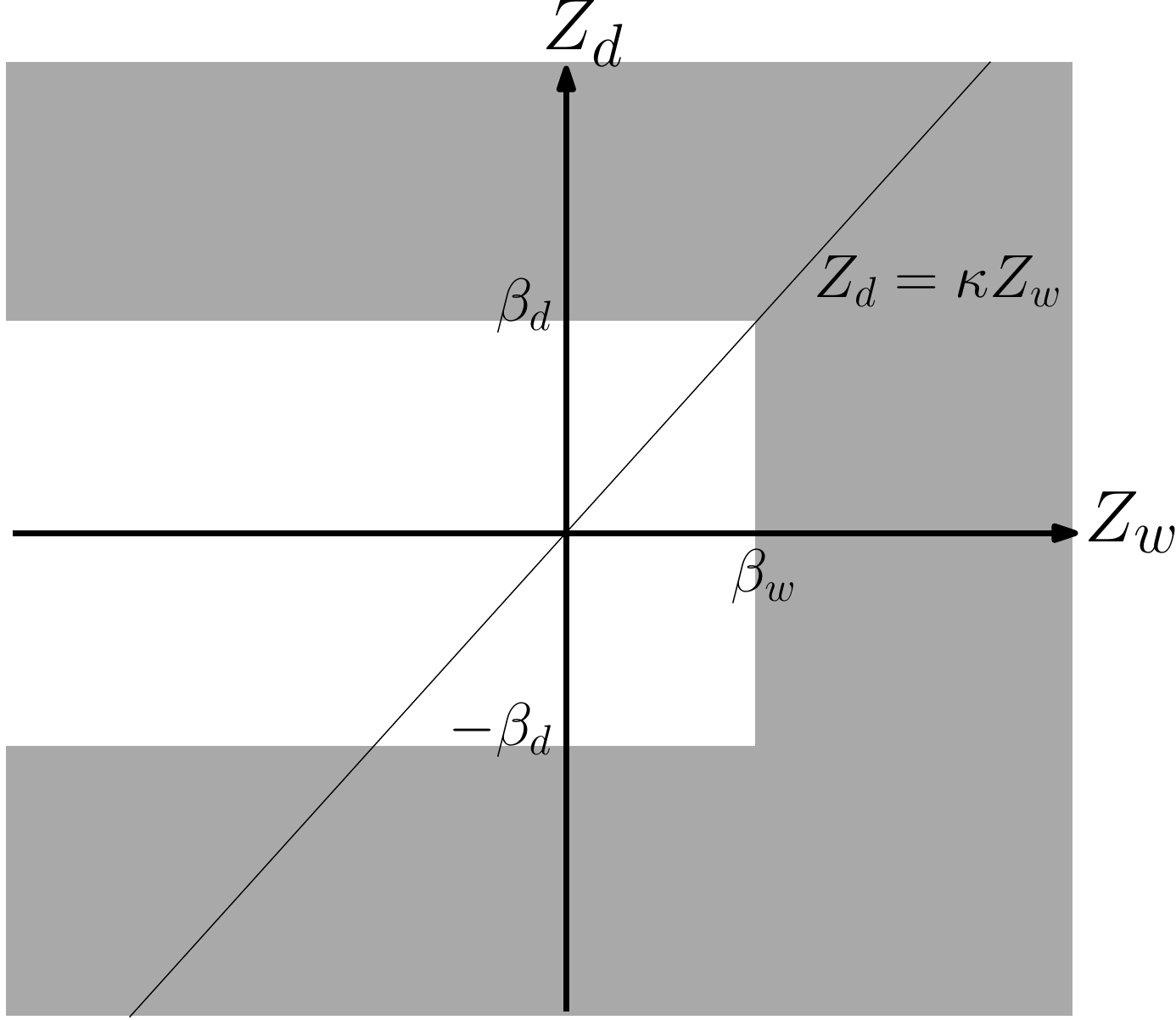}
    \caption{Rejection regions (in gray) of $S,~R_w,~M(\kappa)$. Left: $\{S\geq r_s\}$; middle: $\{Z_w\geq r_w\}$; right: $\{M(\kappa)\geq\beta\}$.}\label{Rej}
\end{figure}

In general, $Z_w$ aims for detecting location alternative and $Z_d$ aims for detecting scale alternative, so the extended generalized edge-count test and the extended max-type edge-count test are effective on both alternatives.  On the other hand, if we know in prior that the difference is only in mean, then the extended weighted edge-count tests are preferred.


\end{remark}

\section{Asymptotics}
\label{sec:asym}
In this section, we provide the asymptotic distributions of new test statistics described in Sections \ref{sec:extended}.  This provides us theoretical bases for obtaining analytic $p$-value approximation.  We then examine how well these approximations work for finite samples.  In the following, we use $a=O(b)$ to denote that $a$ and $b$ are of the same order and $a=o(b)$ to denote that $a$ is of a smaller order than $b$. Let $\e_{i,2}^G$ be the set of edges in $G$ that contain at least one node in $\V_i^G.$

\subsection{Statistics based on averaging} \label{sec:asym_a}
To derive the asymptotic behavior of the statistics based on averaging ($R_{w,(a)}, S_{(a)}, M_{(a)}$), we work under the following conditions:
\begin{condition}\label{C1}
  \begin{equation*}
    |C_0|, \sum_{(u,v)\in C_0}\frac{1}{m_um_v} = O(N), 
  \end{equation*}
  \begin{equation*}
    K, \sum_{u}\frac{1}{m_u}=O(N^{\alpha}), \quad\alpha\leq 1.
  \end{equation*}
\end{condition}
\begin{condition}\label{C2}
  \[
  \sum_{u}m_u(m_u+|\e_u^{C_0}|)(m_u+\sum_{v\in \mathcal{V}_u^{C_0}}m_v+|\e_{u,2}^{C_0}|) = o(N^{3/2}),
  \]
  \[
  \sum_{(u,v)\in C_0}(m_u+m_v+|\e_u^{C_0}|+|\e_v^{C_0}|)(m_u+m_v+\sum_{w\in (\mathcal{V}_u^{C_0}\cup\mathcal{V}_v^{C_0})}m_w+|\e_{u,2}^{C_0}|+|\e_{v,2}^{C_0}|) = o(N^{3/2}).
  \]
\end{condition}
\begin{condition}\label{C3}
\[
\sum_u\frac{(|\e_u^{C_0}|-2)^2}{4m_u}-\frac{(|C_0|-K)^2}{N} = O(N).
\]
\end{condition}
\begin{remark}\label{remark:ave}
One special case for Condition \ref{C1} is 
\begin{equation}\label{remark1-1}
    |C_0|, \sum_{(u,v)\in C_0}\frac{1}{m_um_v}, K, \sum_{u}\frac{1}{m_u}=O(N).
\end{equation}
Conditions \ref{remark1-1} and \ref{C2} are the same conditions stated in \cite{chen2013} in obtaining the asymptotic properties of $R_{0,(a)}$ and $R_{0,(u)}$.  Condition \ref{C1} is easy to be satisfied and Condition \ref{C2} sets constrains on the number of repeated observations and the degrees of nodes in the graph $C_0$ such that they cannot be too large.

When $m_u\equiv m$ for all $u$, Condition \ref{C2} can be simplified to
  \[
  \sum_u|\e_u^{C_0}||\e_{u,2}^{C_0}| = o(N^{3/2})
  \]
  and
  \[
  \sum_{(u,v)\in C_0}(|\e_u^{C_0}|+|\e_v^{C_0}|)(|\e_{u,2}^{C_0}|+|\e_{v,2}^{C_0}|) = o(N^{3/2}).
  \]

  The additional condition (Condition \ref{C3}) makes sure that $(R_1,R_2)^T$ does not degenerate asymptotically. When $m_u\equiv m$ for all $u$, Condition \ref{C3} becomes
\[
\frac{1}{4m}\sum_u|\e_u^{C_0}|^2-\frac{|C_0|^2}{mK} = \frac{1}{4m}\sum_u(|\e_u^{C_0}|-\frac{2|C_0|}{K})^2 = O(N),
\]
which is the variance of the degrees of nodes in $C_0$. When there is not enough variety in the degrees of nodes in $C_0$, the correlation between $R_1$ and $R_2$ tends to 1. (A similar condition is needed for the continuous counterpart \citep{chen2017new}.)
\end{remark}
\begin{theorem}\label{theorem:gauss1}
  Under Conditions \ref{C1}, \ref{C2} and \ref{C3}, as $N\rightarrow\infty$,
 \begin{equation*}
   \begin{pmatrix}
     Z_{w,(a)} \\
     Z_{d,(a)}
   \end{pmatrix}
   \stackrel{D}{\longrightarrow}\mathcal{N}_2(0,\I_2)
 \end{equation*}
 under the permutation null distribution.
\end{theorem}

The proof of this theorem is in supplementary materials. Based on Theorem \ref{theorem:gauss1}, it is easy to obtain the asymptotic distributions of $S_{(a)}$ and $M_{(a)}(\kappa)$.
\begin{corollary}\label{thSa}
   Under Conditions \ref{C1}, \ref{C2} and \ref{C3}, as $N\rightarrow\infty$, $S_{(a)}\stackrel{D}{\longrightarrow}\mathcal{X}_2^2$  under the permutation null distribution.
\end{corollary}
\begin{corollary}\label{th6}
  Under Conditions \ref{C1}, \ref{C2} and \ref{C3}, the asymptotic cumulative distribution function of $M_{(a)}(\kappa)$ is $\Phi(\frac{x}{\kappa})(2\Phi(x)-1)$ under the permutation null distribution, where $\Phi(x)$ denotes the cumulative distribution function of the standard normal distribution.
\end{corollary}

\subsection{Statistics based on taking union} \label{sec:asym_u}
\begin{condition}\label{C4}
  \begin{equation*}
    |\bar{G}|=O(N).
  \end{equation*}
\end{condition}
\begin{condition}\label{C5}
\[
\sum_{i=1}^{N}|\e_i^{\bar{G}}|^2-\frac{4}{N}|\bar{G}|^2=O(N).
\]
\end{condition}
\begin{condition}\label{C6}
  \[
 \sum_{u=1}^{K}m_u^3(m_u+\sum_{v\in\V^{C_0}_u}m_v)\sum_{v\in\{u\}\cup\V^{C_0}_u}m_v(m_v+\sum_{w\in\V^{C_0}_v}m_w)=o(N^{3/2}),
  \]
  \begin{align*}
   \sum_{(u,v)\in C_0} m_um_v\Bigg[m_u(m_u & +\sum_{w\in\V^{C_0}_u}m_w)+ m_v(m_v+\sum_{w\in\V^{C_0}_v}m_w)\Bigg] \\
    & \cdot\Bigg[\sum_{\substack{w\in\{u\}\cup\{v\}\cup\V^{C_0}_u\cup\V^{C_0}_v\\y\in\V^{C_0}_w}}m_w(m_w+m_y)\Bigg] =o(N^{3/2}).
  \end{align*}
\end{condition}
\begin{remark}
Condition \ref{C4} is easy to satisfy. Condition \ref{C5} was mentioned in \cite{chen2017new} in the continuous version. When $m_u\equiv m$ for all $u$, Condition \ref{C5} could be rewritten as
  \[
  \sum_{u=1}^{K}|\e_u^{C_0}|^2-\frac{4}{K}|C_0|^2=O(K).
  \]
  If $C_0$ is the $k$-MST, $k=O(1)$, constructed under Euclidean distance, the above condition always holds based on resultsd in \cite{chen2017new}.
  
  When $m_u\equiv m$ for all $u$, Condition \ref{C6} becomes
  \[
  \sum_u|\e_u^{C_0}||\e_{u,2}^{C_0}| = o(N^{3/2})
  \]
  and
  \[
  \sum_{(u,v)\in C_0}(|\e_u^{C_0}|+|\e_v^{C_0}|)(|\e_{u,2}^{C_0}|+|\e_{v,2}^{C_0}|) = o(N^{3/2}),
  \]
  which are same as the simplified form in Remark \ref{remark:ave}. These conditions restrict the degrees of nodes in graph $C_0$.
\end{remark}
\begin{theorem}\label{theorem:gauss2}
  Under Conditions \ref{C4}, \ref{C5} and \ref{C6}, as $N\rightarrow\infty$,
  \begin{equation*}
   \begin{pmatrix}
     Z_{w,(u)} \\
     Z_{d,(u)}
   \end{pmatrix}
   \stackrel{D}{\longrightarrow}\mathcal{N}_2(0,\I_2),
 \end{equation*}
 under the permutation null distribution.
\end{theorem}

The proof of this theorem is in supplementary materials. Based on Theorem \ref{theorem:gauss2}, it is easy to obtain the asymptotic distributions of $S_{(u)}$ and $M_{(u)}(\kappa)$.
\begin{corollary}\label{thSu}
   Under Conditions \ref{C4}, \ref{C5} and \ref{C6}, as $N\rightarrow\infty$, $S_{(u)}\stackrel{D}{\longrightarrow}\mathcal{X}_2^2$ under the permutation null distribution.
\end{corollary}
\begin{corollary}\label{th6t}
  Under Conditions \ref{C4}, \ref{C5} and \ref{C6}, the asymptotic cumulative distribution function of $M_{(u)}(\kappa)$ is $\Phi(\frac{x}{\kappa})(2\Phi(x)-1)$ under the permutation null distribution, where $\Phi(x)$ denotes the cumulative distribution function of the standard normal distribution.
\end{corollary}

\subsection{Analytic $p$-value approximations}
The asymptotic results in Sections \ref{sec:asym_a} and \ref{sec:asym_u} provide theoretical bases for analytic $p$-values approximations. Here we check how well the analytic $p$-values approximations based on asymptotic results work under finite samples by comparing them with  permutation $p$-values calculated from 10,000 random permutations. 

In the following, we generate data from mechanism (i) in Section \ref{sec:performance} with $\theta_1=\theta_2=5$, $\eta_1=\{1,2,3,4,5\}$ and $\eta_2=\{1,4,3,2,5\}$. We set $C_0$ be the NNL and examine the difference of the asymptotic $p$-value and permutation $p$-value under various settings.


\begin{figure}[h]
  \centering
  \includegraphics[width=0.84\textwidth]{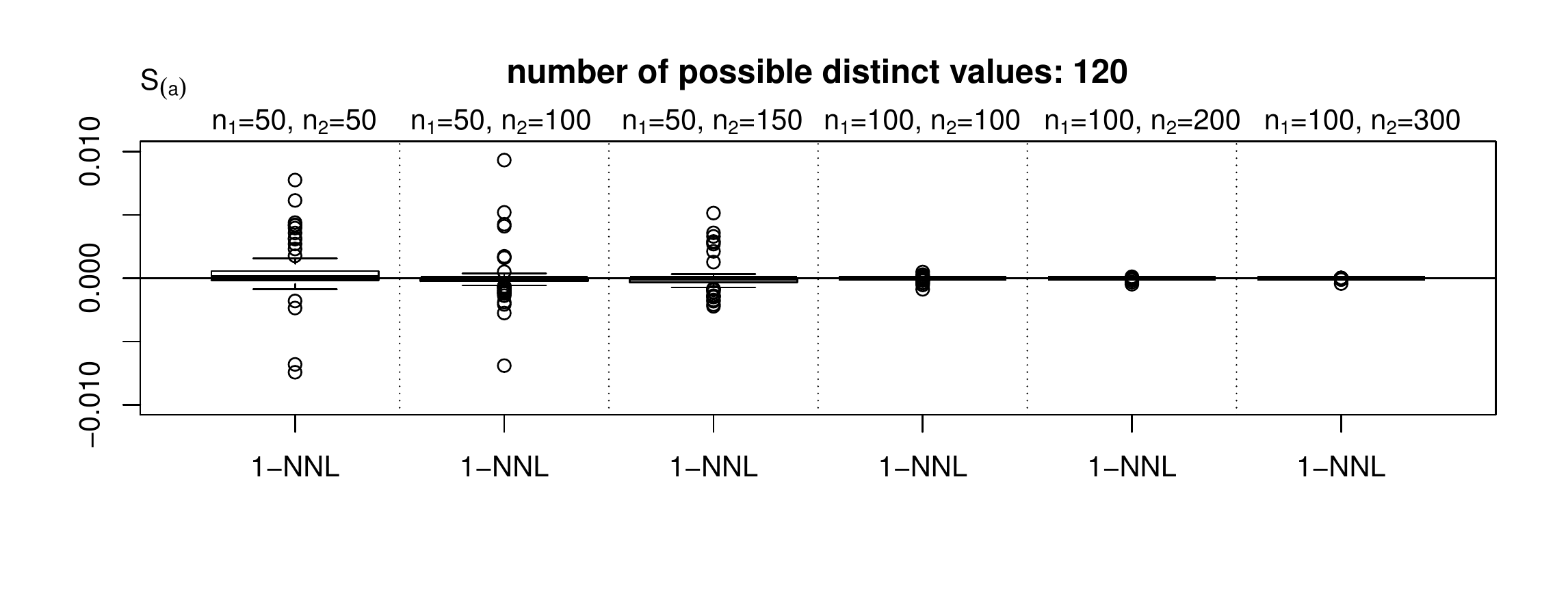}
 
 \vspace{-2.2em}
  
  \includegraphics[width=0.84\textwidth]{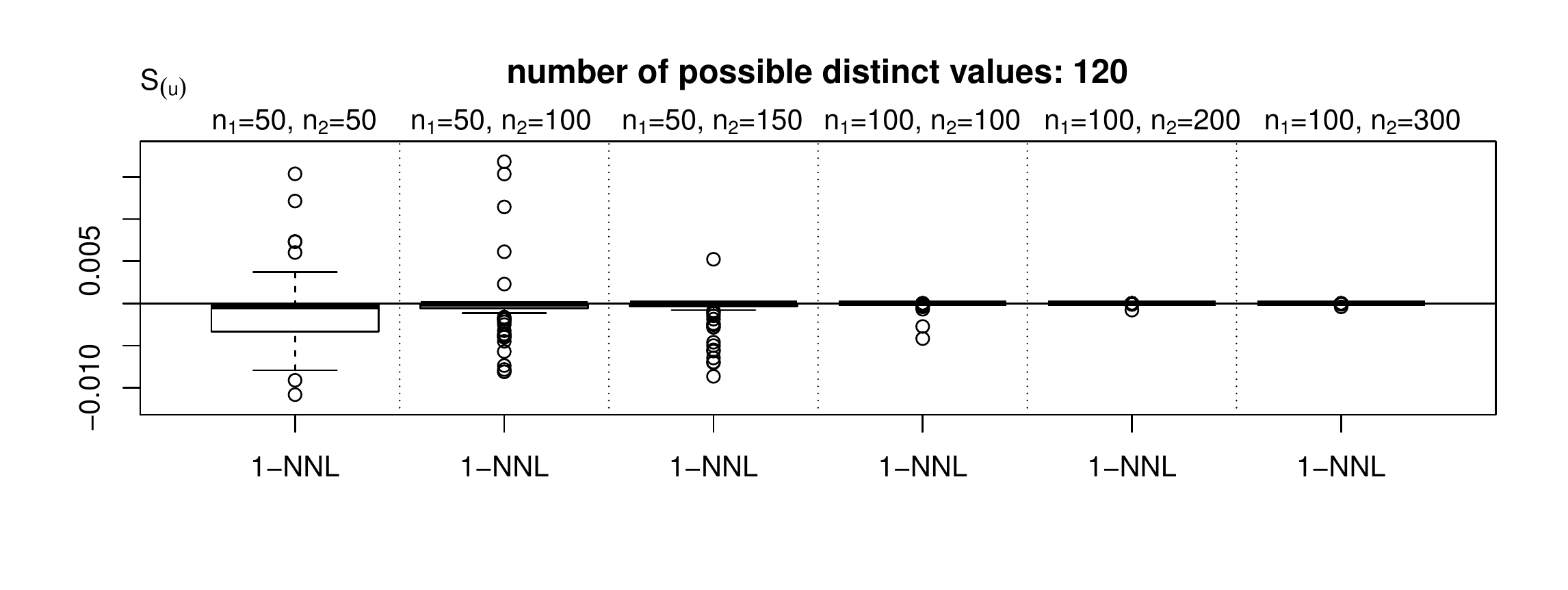}

 \vspace{-2.2em}
  
  \includegraphics[width=0.84\textwidth]{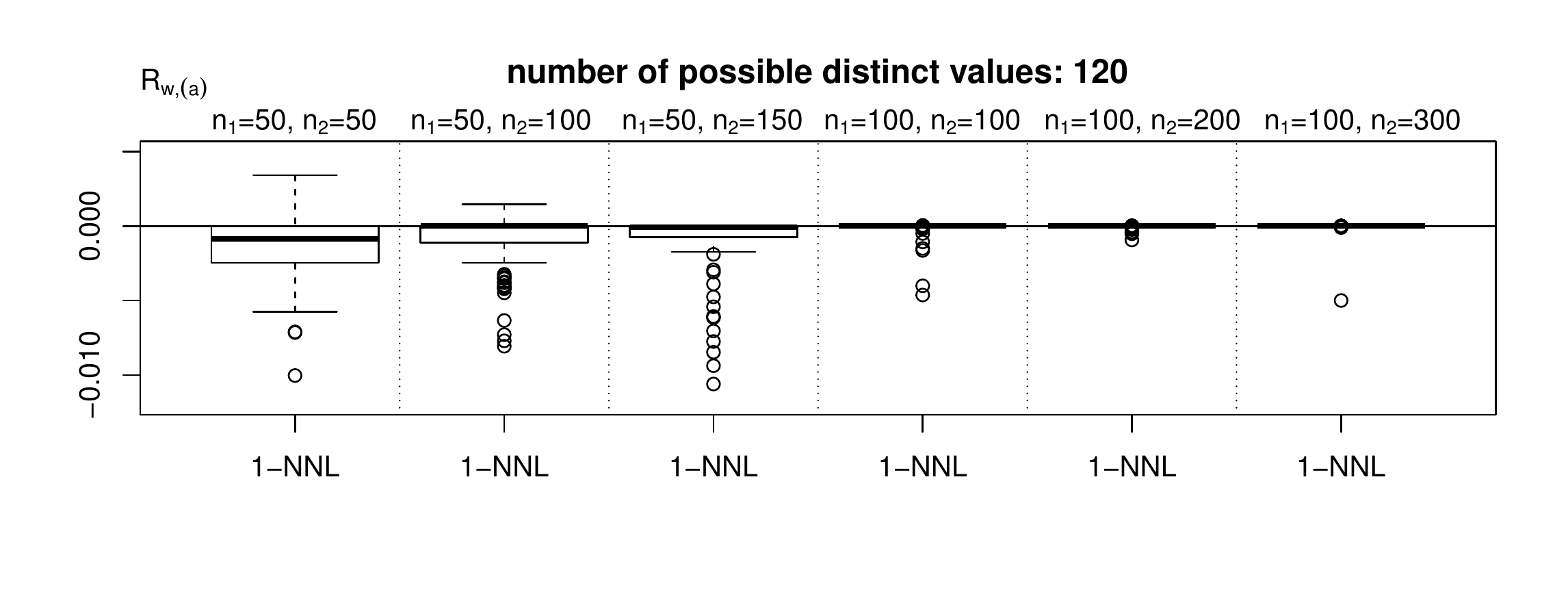}

 \vspace{-2.2em}
  
  \includegraphics[width=0.84\textwidth]{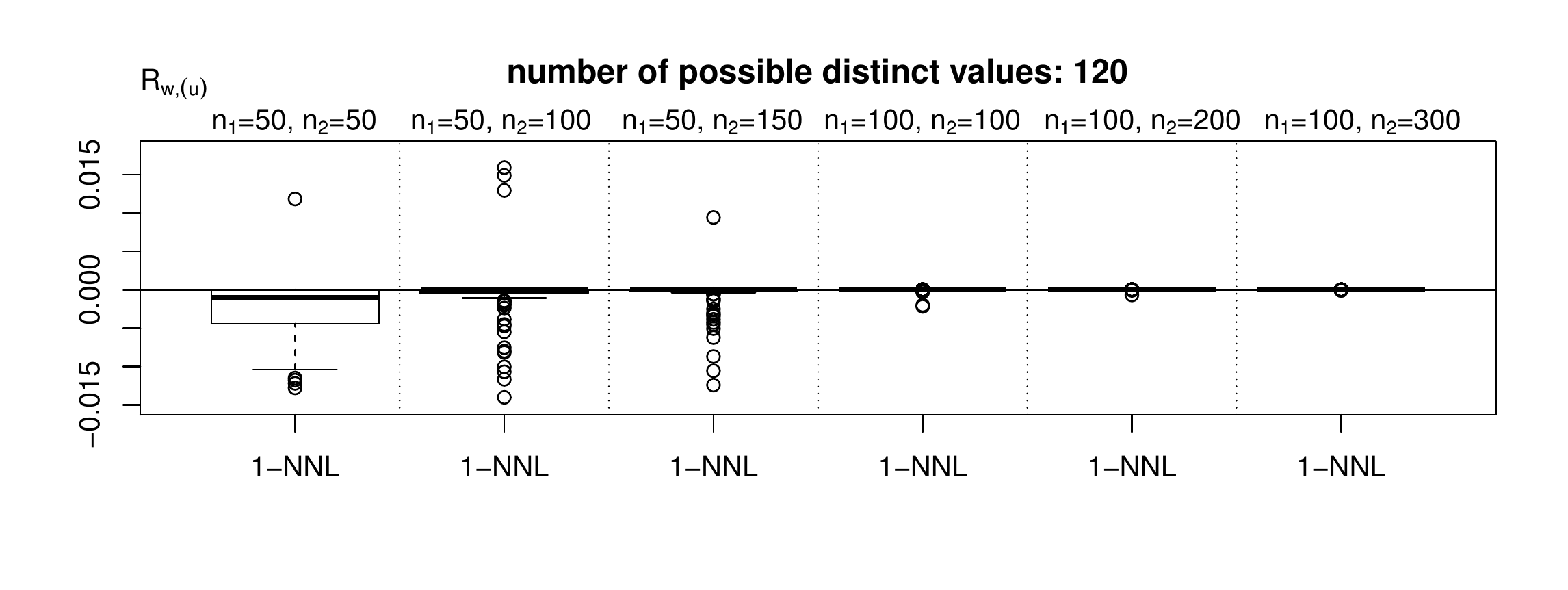}
 
 \vspace{-2.2em}
 
 \caption{Boxplots for the differences between the asymptotic $p$-value and the permutation $p$-value based on 100 simulation runs under each setting for $S_{(a)}, S_{(u)}, R_{w,(a)}$ and $R_{w,(u)}$.}\label{p-value1}
\end{figure}

\begin{figure}[h]
  \centering
  \includegraphics[width=0.84\textwidth]{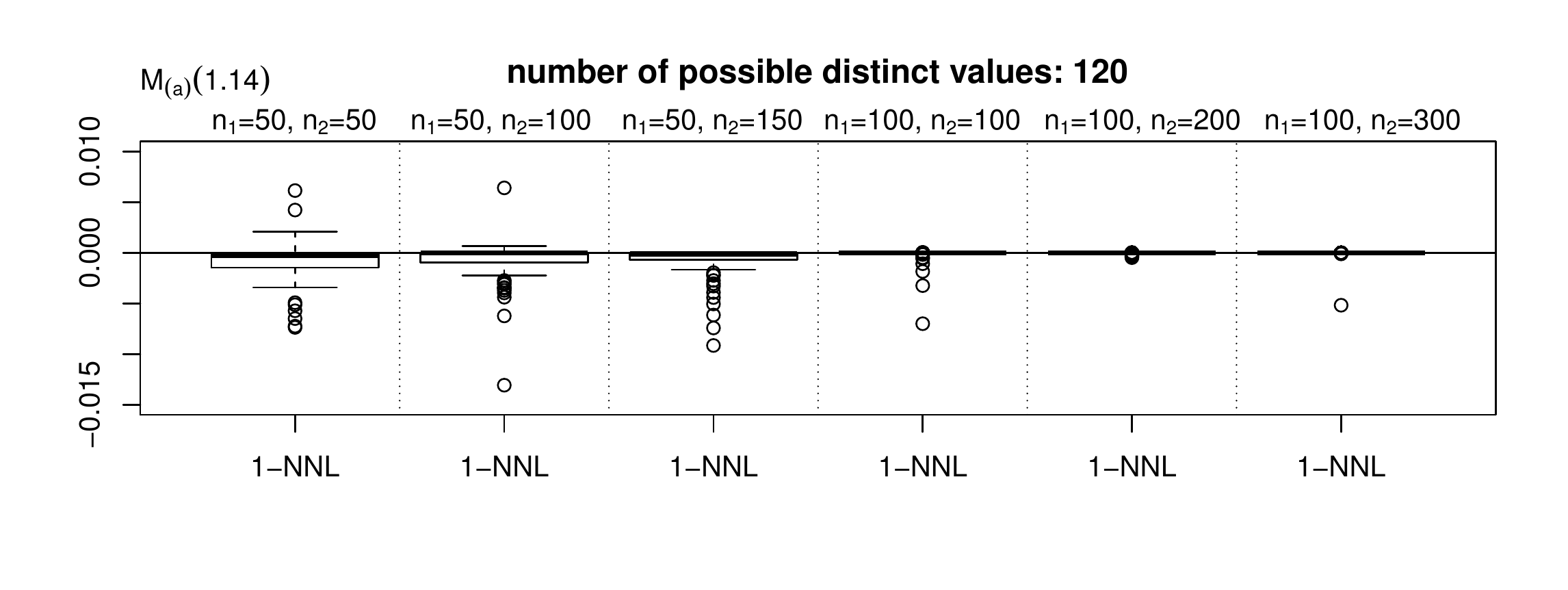}
 
 \vspace{-2.2em}
  
  \includegraphics[width=0.84\textwidth]{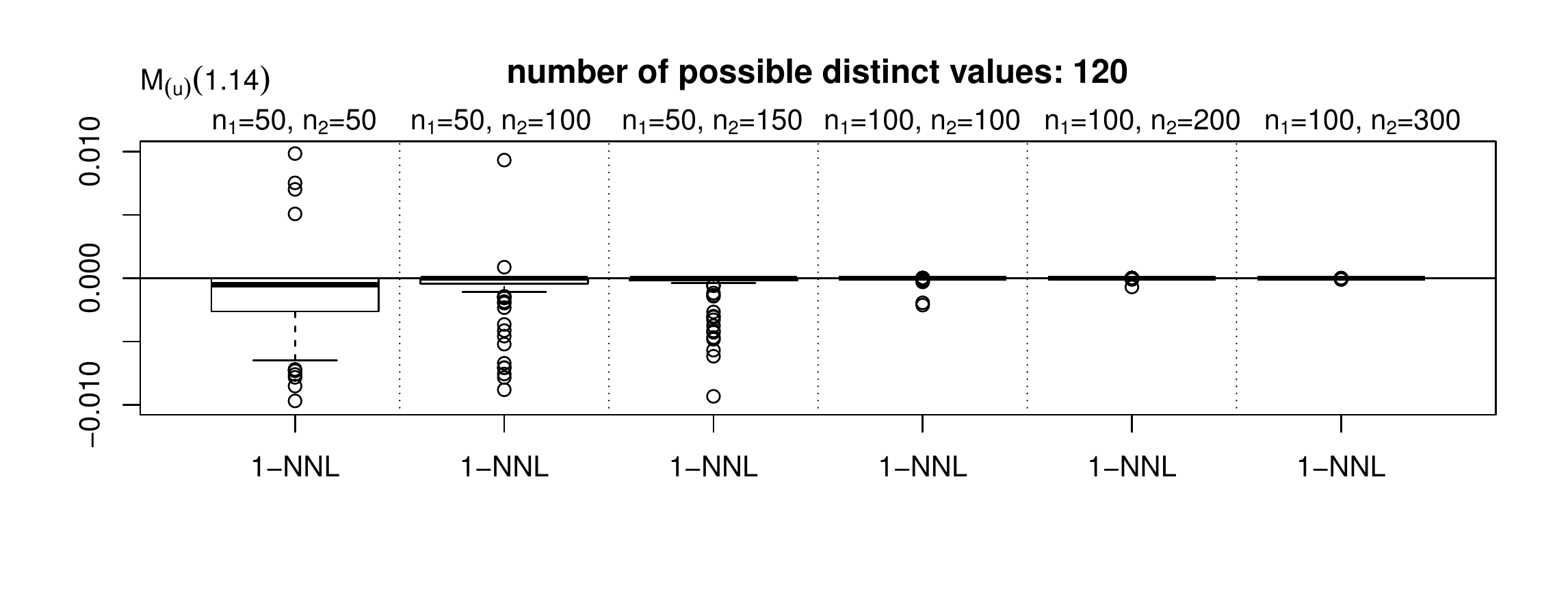}
 
 \vspace{-2.2em}
  
 \caption{Boxplots for the differences between the asymptotic $p$-value and the permutation $p$-value based on 100 simulation runs under each setting for $M_{(a)}(1.14)$ and $M_{(u)}(1.14)$.}\label{p-value2}
\end{figure}

%

Figures \ref{p-value1} and \ref{p-value2} show boxplots for the differences of the two $p$-values (asymptotic $p$-value minus permutation $p$-value) with different choices of $n_1$ and $n_2$ for $S_{(a)}, S_{(u)}, R_{w,(a)}, R_{w,(u)}, M_{(a)}(1.14)$ and $M_{(u)}(1.14)$.  (The results for $M_{(a)}(\kappa)$ and $M_{(u)}(\kappa)$ for $\kappa=1.31, 1$ are similar to those with $\kappa=1.14$ and are shown in supplementary materials.)  We see that when both $n_1$ and $n_2$ are over 100, the asymptotic $p$-value is very close to the permutation $p$-value for all new test statistics.

\section{Phone-call network data analysis}
\label{sec:network}
In this section, we analyze the phone-call network data mentioned in Section \ref{sec:intro} in details. We first present the test results of various statistics, and then examine the analytic $p$-value approximations through this real data example.

The MIT Media Laboratory conducted a study following 106 subjects, including students and staffs in an institute, who used mobile phones with pre-installed software that can record call logs. The study lasted from July 2004 to June 2005 (\cite{eagle2009}). Given the richness of this dataset, many problems can be studied. One question of interest is whether phone call patterns on weekdays are different from those on weekends. The phone calls on weekdays and weekends can be viewed as representations of professional relationship and personal relationship, respectively.

We bin the phone calls by day and, for each day, construct a directed phone-call network with the 106 subjects as nodes and a directed edge pointing from person $i$ to person $j$ if person $i$ made one call to person $j$ on that day. We encode the directed network of each day by an adjacency matrix, with 1 for element $[i,j]$ if there is a directed edge pointing from subject $i$ to subject $j$, and 0 otherwise.

The original dataset was sorted in the calendar order with 236 weekdays and 94 weekends.
Among the 330 (236+94) networks, 
there are 285 distinct values and 11 of them have more than one observations. We denote the distinct values as matrices $B_1,\cdots,B_{285}$. 
We adopt the distance measure used in \cite{chen2017new} and \cite{chen2018weighted}, which is defined as the number of different entries, i.e., 
\[
 d(B_i,B_j)=\|B_i-B_j\|_F^2, 
\]
where $\|\cdot\|_F$ is the Frobenius norm of a matrix.
Besides the repeated observations, there are many equal distances among distinct values.  We set $C_0$ to be the 3-NNL, which has similar density as the 9-MST recommended in \cite{chen2018weighted}.

\begin{table}[!h]
  \centering
  \caption{Breakdown statistics of the phone-call network data.}\label{AppT1}
  \begin{tabular}{c|cccc}
    \hline\hline
     & Value & Mean & Value-Mean & SD \\ \hline\hline
    $R_{1,(a)}$ & 2800.26 & 2669.56 & 130.70 & 143.33 \\ \hline
    $R_{2,(a)}$ & 409.18 & 420.80 & -11.62 & 57.75 \\ \hline
    $(R_{1,(a)}+R_{2,(a)})/2$ & 1604.72 & 1545.18 & 59.54 & 44.74 \\ \hline
    $R_{w,(a)}$ & 1087.14 & 1058.40 & 28.73 & 11.79\\ \hline
    $R_{d,(a)}$ & 2391.08 & 2248.76 & 142.32 & 199.37 \\ \hline\hline
  \end{tabular}

  \vspace{0.3cm}
  \begin{tabular}{c|cccc}
    \hline\hline
     & Value & Mean & Value-Mean & SD \\ \hline\hline
    $R_{1,(u)}$ & 7163.00 & 6860.35 & 302.65 & 381.50\\ \hline
    $R_{2,(u)}$ & 1008.00 & 1081.38  & -73.38 & 151.66 \\ \hline
    $(R_{1,(u)}+R_{2,(u)})/2$ & 4085.50 & 3970.86 & 114.64 & 116.22 \\ \hline
    $R_{w,(u)}$ & 2753.17 & 2719.93  & 33.24 & 15.65\\ \hline
    $R_{d,(u)}$ & 6155.00 & 5778.97 &  376.03 & 532.03 \\ \hline\hline
  \end{tabular}

  \vspace{0.3cm}
  \begin{tabular}{c|c|cc||c|c|cc}
    \hline\hline
    \multicolumn{2}{c|}{} & Value & $p$-Value & \multicolumn{2}{c|}{} & Value & $p$-Value \\ \hline
    \multicolumn{2}{c|}{$Z_{0,(a)}$} & -1.33 & 0.092 & \multicolumn{2}{c|}{$Z_{0,(u)}$} & -0.99 & 0.162 \\ \hline
    \multicolumn{2}{c|}{$S_{(a)}$} & 6.45 & 0.040 & \multicolumn{2}{c|}{$S_{(u)}$} & 5.01 & 0.082 \\ \hline
    \multicolumn{2}{c|}{$Z_{w,(a)}$} & 2.44 & {0.007} & \multicolumn{2}{c|}{$Z_{w,(u)}$} & 2.12 & {0.017} \\ \hline
    \multicolumn{2}{c|}{$|Z_{d,(a)}|$} & 0.71 & 0.475 & \multicolumn{2}{c|}{$|Z_{d,(u)}|$} & 0.71 & 0.480 \\ \hline
    \multirow{3}*{$M_{(a)}(\kappa)$} & $\kappa=1.31$ & 3.19 & 0.009 & \multirow{3}*{$M_{(u)}(\kappa)$} & $\kappa=1.31$ & 2.78 & 0.022 \\ \cline{2-4}\cline{6-8}
    & $\kappa=1.14$ & 2.78 & 0.013 & & $\kappa=1.14$ & 2.42 & 0.032 \\ \cline{2-4}\cline{6-8}
    & $\kappa=1$ & 2.44 & 0.022 & & $\kappa=1$ & 2.12 & 0.050 \\
    \hline\hline
  \end{tabular}
\end{table}

Table \ref{AppT1} lists the results. In particular, we list the values, expectation (Mean) and standard deviations (SD) of $R_{1,(a)}$, $R_{1,(u)}$, $R_{2,(a)}$, $R_{2,(u)}$, $(R_{1,(a)}+R_{2,(a)})/2$, $(R_{1,(u)}+R_{2,(u)})/2$, $R_{w,(a)}$, $R_{w,(u)}$, $R_{d,(a)}$ and $R_{d,(u)}$, as well as the values and $p$-values of $Z_{0,(a)},~Z_{0,(u)},$ $S_{(a)},~S_{(u)},$ $Z_{w,(a)},~Z_{w,(u)},$ $|Z_{d,(a)}|,~|Z_{d,(u)}|$, $M_{(a)}(\kappa)$ and $M_{(u)}(\kappa)$, where $Z_{0,(a)}$ and $Z_{0,(u)}$ are standardizations for $R_{0,(a)}$ and $R_{0,(u)}$, respectively. Note that the tests based on $(R_{1,(a)}+R_{2,(a)})/2$, and $(R_{1,(u)}+R_{2,(u)})/2$ are equivalent to those based on $R_{0,(a)}$ and $R_{0,(u)}$, respectively.

 We first check results based on ``averaging''. We can see that $R_{1,(a)}$ is much higher than its expectation, while $R_{2,(a)}$ is smaller than its expectation.
 The original edge-count test $R_{0,(a)}$ is equivalent to adding $R_{1,(a)}$ and $R_{2,(a)}$ directly, so the signal in $R_{1,(a)}$ is diluted by $R_{2,(a)}$.  In addition, due to the variance boosting issue, it fails to reject the null hypothesis at 0.05 significance level. On the other hand, the weighted edge-count test chooses the proper weight to minimize the variance and performs well. Since $S_{(a)}$ and $M_{(a)}(\kappa)$ consider the weighted edge-count statistic and the difference of two with-in sample edge-counts simultaneously, these tests all reject the null at 0.05 significance level. The larger the $\kappa$ is, the more similar the max-type test ($M_{(a)}(\kappa)$) and the weighted test ($R_{w,(a)}$) are. So the $p$-values of $M_{(a)}(\kappa)$ are very close to that of $R_{w,(a)}$, when $\kappa$ is large. The results on the ``union'' counterparts are similar, except that $S_{(u)}$ cannot reject the null at 0.05 significance level.  Based on the information in the table, it is clear that there is mean difference between the two samples, while no significant scale difference between the two samples.
 
 We also check the analytic $p$-values obtained based on asymptotical results with those based on 10,000 random permutations and the results are shown in Table \ref{pvalue:real1}.  We can see that the asymptotic $p$-values and the permutation $p$-values are quite close for all test statistics.

\begin{table}[!h]
  \centering
  \caption{The $p$-value obtained from the asymptotic results (Asym.) and from doing 10,000 random permutations (Perm.) for different statistics.}\label{pvalue:real1}
  \vspace{0.3cm}
  \begin{tabular}{|c|cc||c|cc|}
    \hline
    $p$-value & Asym. & Perm. & $p$-value & Asym. & Perm.  \\ \hline
    $S_{(a)}$ & 0.040 &  0.042 & $S_{(u)}$ & 0.082 & 0.086 \\
    $R_{w,(a)}$ & 0.007 & 0.013 & $R_{w,(u)}$ & 0.017 & 0.024 \\
    $M_{(a)}(1.31)$ & 0.009 &  0.014 & $M_{(u)}(1.31)$ & 0.022 & 0.026 \\
    $M_{(a)}(1.14)$ & 0.013 &  0.019 & $M_{(u)}(1.14)$ & 0.032 & 0.034\\
    $M_{(a)}(1)$ & 0.022 & 0.025 & $M_{(u)}(1)$ & 0.050 & 0.049 \\
    \hline
  \end{tabular}
\end{table}

\section{Conclusion}
\label{sec:discussion}
The generalized edge-count test and the weighted edge-count test are useful tools in two-sample testing regime. Both tests rely on a similarity graph constructed on the pooled observations from the two samples and can be applied to various data types as long as a reasonable similarity measure on the sample space can be defined. However, they are problematic when the similarity graph is not uniquely defined, which is common for data with repeated observations. In this work, we extend them as well as a max-type statistic, to accommodate scenarios when the similarity graph cannot be uniquely defined. The extended test statistics are equipped with easy-to-evaluate analytic expressions, making them easy to compute in real data analysis. The asymptotic distributions of the extended test statistics are also derived and simulation studies show that the $p$-values obtained based on asymptotic distributions are quite accurate under sample sizes in hundreds and beyond, making these tests easy-off-the-shelf tools for large data sets.

Among the extended edge-count tests, the extended weighted edge-count tests aim for location alternatives, and the extended generalized/max-type edge-count tests aim for more general alternatives.  When these tests do not reach a consensus, a detailed analysis illustrated by the phone-call network data in Section \ref{sec:network} is recommended.

\begin{supplement}
\sname{Supplement to ``Graph-based two-sample tests for data with repeated observations"} 
\slink[url]{http://www.e-publications.org/ims/support/dowload/imsart-ims.zip}
\sdescription{The supplementary material contains proofs of lemmas and theorems, and some additional results.}
\end{supplement}

\section*{Acknowledgments}
Jingru Zhang is supported in part by the CSC scholarship. Hao Chen is supported in part by NSF award DMS-1513653.

\bibliographystyle{imsart-nameyear}
\bibliography{test}

\appendix

\section{Analytic expressions of the expectation and variance for the extended basic quantities} \label{sec:AppA}
\begin{lemma}\label{lem1}
  The means, variances and covariance of $R_{1,(a)}$ and $R_{2,(a)}$ under the permutation null are
  \begin{flalign*}
  \begin{split}
  & \qquad\EP(R_{1,(a)}) = (N-K+|C_0|)p_1, \\
 & \qquad\EP(R_{2,(a)}) = (N-K+|C_0|)q_1,
 \end{split}&
  \end{flalign*}
  \vspace{-0.5cm}
  \begin{flalign*}
  \begin{split}
    \qquad\varP(R_{1,(a)}) = & 4(p_2-p_3)(N-K+2|C_0|+\sum_{u}\frac{|\e_u^{C_0}|^2}{4m_u}-\sum_{u}\frac{|\e_u^{C_0}|}{m_u})\\
  & +(p_3-p_1^2)(N-K+|C_0|)^2+(p_1-2p_2+p_3)\sum_{(u,v)\in C_0}\frac{1}{m_um_v}\\
  & +2(p_1-4p_2+3p_3)(K-\sum_{u}\frac{1}{m_u}),
  \end{split}&
  \end{flalign*}
  \vspace{-0.5cm}
  \begin{flalign*}
  \begin{split}
     \qquad\varP(R_{2,(a)}) = & 4(q_2-q_3)(N-K+2|C_0|+\sum_{u}\frac{|\e_u^{C_0}|^2}{4m_u}-\sum_{u}\frac{|\e_u^{C_0}|}{m_u})\\
  & +(q_3-q_1^2)(N-K+|C_0|)^2+(q_1-2q_2+q_3)\sum_{(u,v)\in C_0}\frac{1}{m_um_v}\\
  & +2(q_1-4q_2+3q_3)(K-\sum_{u}\frac{1}{m_u}),
  \end{split}&
  \end{flalign*}
  \vspace{-0.5cm}
  \begin{flalign*}
  \begin{split}
     \qquad\covP(R_{1,(a)},R_{2,(a)}) = & (f_1-p_1q_1)(N-K+|C_0|)^2 \\
    & + f_1\Bigg[-4(N-K+2|C_0|+\sum_{u}\frac{|\e_u^{C_0}|^2}{4m_u}-\sum_{u}\frac{|\e_u^{C_0}|}{m_u}) \\
    & + 6(K-\sum_{u}\frac{1}{m_u}) + \sum_{(u,v)\in C_0}\frac{1}{m_um_v}\Bigg],
    \end{split}&
  \end{flalign*}
  where
  \[
  p_1 = \frac{n_1(n_1-1)}{N(N-1)},\quad p_2 = \frac{n_1(n_1-1)(n_1-2)}{N(N-1)(N-2)},\quad p_3 = \frac{n_1(n_1-1)(n_1-2)(n_1-3)}{N(N-1)(N-2)(N-3)},
  \]
  \[
  q_1 = \frac{n_2(n_2-1)}{N(N-1)},\quad q_2 = \frac{n_2(n_2-1)(n_2-2)}{N(N-1)(N-2)},\quad q_3 = \frac{n_2(n_2-1)(n_2-2)(n_2-3)}{N(N-1)(N-2)(N-3)},
  \]
  \[
  f_1 = \frac{n_1(n_1-1)n_2(n_2-1)}{N(N-1)(N-2)(N-3)}.
  \]
\end{lemma}
Lemma \ref{lem1} is proved in supplementary materials..

\begin{lemma}\label{lem2}
  The means, variances and covariance of $R_{1,(u)}$ and $R_{2,(u)}$ under the permutation null are
  \begin{align*}
 & \EP(R_{1,(u)}) = |\bar{G}|p_1,\\
 & \EP(R_{2,(u)}) = |\bar{G}|q_1,\\
 & \varP(R_{1,(u)}) = (p_1-p_3)|\bar{G}|+(p_2-p_3)\sum_{i=1}^{N}|\e_i^{\bar{G}}|(|\e_i^{\bar{G}}|-1)+(p_3-p_1^2)|\bar{G}|^2,\\
 & \varP(R_{2,(u)}) = (q_1-q_3)|\bar{G}|+(q_2-q_3)\sum_{i=1}^{N}|\e_i^{\bar{G}}|(|\e_i^{\bar{G}}|-1)+(q_3-q_1^2)|\bar{G}|^2,\\
  & \covP(R_{1,(u)},R_{2,(u)}) = f_1\left[|\bar{G}|^2-|\bar{G}|-\sum_{i=1}^{N}|\e_i^{\bar{G}}|(|\e_i^{\bar{G}}|-1)\right]-p_1q_1 |\bar{G}|^2.
  \end{align*}
where $p_1,p_2,p_3,q_1,q_2,q_3,f_1$ are defined as those in Lemma \ref{lem1}.
\end{lemma}
Lemma \ref{lem2} is proved in supplementary materials.

\end{document}